\documentclass[aps,prx,10pt,floatfix,twocolumn,superscriptaddress,longbibliography]{revtex4-2}

\usepackage{amsmath,amssymb,amsthm,mathtools}
\mathtoolsset{showonlyrefs}

\usepackage{physics}
\usepackage{graphicx}
\usepackage{tikz-cd}
\usepackage{booktabs}
\usepackage{xcolor}
\usepackage{listings,bm}
\usepackage{soul}

\definecolor{magmaKeyword}{RGB}{0,51,102}
\definecolor{magmaComment}{gray}{0.45}
\definecolor{magmaString}{RGB}{102,51,0}
\definecolor{magmaRule}{gray}{0.80}

\lstdefinelanguage{Magma}{
  morekeywords={
    and,or,not,if,then,elif,else,end,for,in,do,while,repeat,until,
    function,procedure,return,break,continue,
    true,false,
    div,mod,
    Matrix,DiagonalMatrix,IdentityMatrix,sub,GL,CyclotomicField,LCM,IsFinite,Isqrt,printf,
    Group,List,DiagonalMat,E
  },
  sensitive=true,
  morecomment=[l]{//},
  morestring=[b]",
}

\lstdefinestyle{magma}{
  language=Magma,
  basicstyle=\small\ttfamily,
  breaklines=true,
  breakatwhitespace=true,
  columns=fullflexible,
  keepspaces=true,
  showstringspaces=false,
  numbers=left,
  numberstyle=\tiny\color{magmaComment},
  frame=single,
  rulecolor=\color{magmaRule},
  framesep=3pt,
  xleftmargin=0.6cm,
  keywordstyle=\color{magmaKeyword},
  commentstyle=\color{magmaComment},
  stringstyle=\color{magmaString}
}

\usepackage{silence}
\usepackage[colorlinks=true,linkcolor=black,citecolor=red,urlcolor=blue]{hyperref}

\newtheorem{theorem}{Theorem}
\newtheorem{proposition}[theorem]{Proposition}
\newtheorem{lemma}[theorem]{Lemma}
\newtheorem{corollary}[theorem]{Corollary}
\newtheorem{definition}[theorem]{Definition}
\newtheorem{remark}[theorem]{Remark}

\newcommand{\SU}{\mathrm{SU}}
\newcommand{\U}{\mathrm{U}}
\newcommand{\SL}{\mathrm{SL}}
\newcommand{\GL}{\mathrm{GL}}
\DeclareMathOperator{\Ad}{Ad}
\DeclareMathOperator{\Span}{span}

\begin{document}


\title{Quantum Universality in Composite Systems: A Trichotomy of Clifford Resources}

\author{Alejandro Borda}
\email{a.borda@uniandes.edu.co}
\affiliation{Department of Physics, Universidad de los Andes, Bogotá, D.C. 111711, Colombia}
\affiliation{Departamento de Matemáticas, Universidad de los Andes, Bogotá, D.C. 111711, Colombia}

\author{Julián Rincón}
\email{julian.rincon@uniandes.edu.co}
\affiliation{Department of Physics, Universidad de los Andes, Bogotá, D.C. 111711, Colombia}

\author{César Galindo}
\email{cn.galindo1116@uniandes.edu.co}
\affiliation{Departamento de Matemáticas, Universidad de los Andes, Bogotá, D.C. 111711, Colombia}

\date{\today}

\begin{abstract}
We show that single-qudit universality in Clifford-based gate sets follows a trichotomy determined by the prime factorization of the local dimension $d$. For prime $d$, any gate outside the Clifford group is universal. For prime-power dimensions $d=p^m$ with $m\ge 2$, not every non-Clifford gate is universal, but it can be achieved by suitable members of a family of diagonal phase gates, generalizing the qubit $T$ gate, as well as by permutations as simple as swapping $\ket{0}$ and $\ket{1}$ while leaving all other basis states unchanged. When $d$ decomposes into pairwise coprime prime powers, generalized CNOT-type gates between the corresponding factors already suffice for universality. In this composite case, universality can be obtained without introducing an explicit diagonal magic gate. Our results split non-Clifford resources for high-dimensional systems into two broad mechanisms: CNOT-type (permutations) or $T$-type (diagonal phases) gates.
\end{abstract}

\maketitle

\section{Introduction}\label{sec:introduction}

 A basic question in quantum information science is how quantum computation departs from classical computation. While quantum entanglement is often cited as one source of this advantage, the Gottesman-Knill theorem shows a limitation: within the stabilizer setting, Clifford circuits can be efficiently simulated classically~\cite{Gottesman1998, AaronsonGottesman2004}. This result extends to higher-dimensional qudits, where efficient classical simulation of Clifford circuits has also been established~\cite{Gheorghiu2014, deBeaudrap2013}. The missing resource that lifts computation beyond classical reach is typically identified as \emph{magic}: non-stabilizer states or non-Clifford gates that enable quantum universality---the ability to approximate any unitary operator to arbitrary precision by finite gate compositions~\cite{kitaevbook}. Establishing conditions under which the Clifford group in high-dimensional systems yields universal quantum computation remains an open problem. In this work, we analyze this question for the single-qudit Clifford group.

Our main focus is \emph{single-qudit universality}. Once the local problem is solved, a theorem of Brylinski and Brylinski, recalled in Theorem~\ref{thm:brylinski}, shows that combining any single-qudit universal set with an entangling two-qudit gate yields universal quantum computation on $n$ qudits~\cite{Brylinski2002}. The present work therefore isolates the local obstruction that underlies multi-qudit universality for arbitrary qudit dimension.

\textbf{Related work.} The structure of universality for qubit systems is well understood: the Clifford group supplemented by any non-Clifford gate---typically the $T$-gate from the third level of the Clifford hierarchy~\cite{GottesmanChuang1999}---generates a dense subgroup of $\SU(2)$~\cite{Gottesman1999, Nebe2001}. This ``Clifford~+~$T$'' paradigm underlies modern fault-tolerant quantum computation, where magic states and gate teleportation supply the non-Clifford operations~\cite{Bravyi2005, Knill2004}.

For higher-dimensional systems, the situation is less complete. Universality criteria for qudits of prime dimension have been established~\cite{Howard2012, Campbell2012}, and the maximality of the Clifford group in these cases follows from classical results on finite linear groups~\cite{Nebe2001}. However, prime-power and composite dimensions have not been systematically studied within a Clifford-based framework, even as experimental platforms, of so-called high-dimensional quantum systems, increasingly access non-binary levels~\cite{RingbauerEtAl2022, Low2025, Chicco_2023}.

Universality for arbitrary dimension $d$, \emph{not} based on the Clifford group, has been established through both constructive and algebraic frameworks. Previous work proved universality by decomposing elements of $\SU(d^n)$ into two-level unitaries, which were further expressed through phase and controlled operations~\cite{Luo2014}. Alternative constructions employed the generators of $\mathfrak{su}(d)$, where universal gate sets were generated via exponentiation of string operators, to produce $\U(d^n)$~\cite{VlasovJMP2002, VlasovSPIE2003}. These results show that universality in arbitrary dimension is known in general, but they do not resolve the specific question studied here: which additional local qudit gates make the Clifford group universal, and how the answer depends on the arithmetic structure of $d$.

Establishing universality amounts to proving that the generated group by such a gate set is dense in $\SU(d)$, the group of unitary operators with determinant one. We address this topological problem via the \emph{adjoint representation} of the group on $\mathfrak{sl}(d,\mathbb{C})$. A finitely generated subgroup of $\SU(d)$ is dense if and only if it is infinite and acts irreducibly on the Lie algebra $\mathfrak{sl}(d,\mathbb{C})$ (Theorem~\ref{thm:density_criterion}). The extent to which the Clifford group satisfies these conditions is governed by the prime factorization $d = p_1^{m_1} \cdots p_k^{m_k}$ (Theorem~\ref{thm:prime_irreducibility} and Theorem~\ref{thm:algebra_stratification}). This factorization dictates the algebraic obstruction and leads to the trichotomy studied here.

\subsection{Main Results}

We summarize the main results of this work as follows:

\begin{enumerate}
    \item[\textbf{(I)}]
    \textbf{Prime dimensions} $\bm{(d = p).}$ No obstruction. The projective Clifford group acts irreducibly on the Lie algebra, and its finite lift is maximal modulo scalar phases. We prove that both maximality and irreducibility hold \emph{if and only if} $d$ is prime: adding \emph{any} non-Clifford gate forces universality (Theorem~\ref{thm:prime_irreducibility}, Theorem~\ref{thm:maximality_iff}, and Corollary~\ref{cor:prime_universality}).

    \item[\textbf{(II)}]
    \textbf{Prime-power dimensions} $\bm{(d = p^m, \, m \geq 2).}$ Reducibility obstruction. The adjoint representation becomes reducible, decomposing the Lie algebra into invariant subspaces indexed by divisibility. The Clifford group acts transitively within each subspace but cannot connect them. Universality therefore requires a non-Clifford gate that restores irreducibility by coupling such subspaces. Two classes of gates accomplish this:
    \begin{itemize}
        \item Diagonal gates from the $T_s$ family with $s \nmid d$, generalizing the qubit $T$ gate to $d$ dimensions.
        
        \item Simple permutations, such as the transposition exchanging $\ket{0}$ and $\ket{1}$ while fixing the remaining basis states.
    \end{itemize}
    Not every non-Clifford gate suffices; one must choose a gate with the correct algebraic properties (Theorem~\ref{thm:prime_power_universality} and Theorem~\ref{thm:transposition_universality}).

    \item[\textbf{(III)}] \textbf{Composite dimensions with coprime factors} $\bm{(d = d_1 \cdots d_k).}$ Obstruction resolved by arithmetic. When the factors are pairwise coprime prime powers, the Clifford group decomposes as a direct product $\mathcal{C}_d \cong \mathcal{C}_{d_1} \times \cdots \times \mathcal{C}_{d_k}$ of the local Clifford groups, and consequently fails to act irreducibly on the global Lie algebra (Proposition~\ref{prop:clifford_decomposition}). However, generalized intra-qudit CNOTs connecting coprime factors automatically restore irreducibility and construct the necessary diagonal phase gates (Proposition~\ref{prop:bezout_construction}). When $k=2$, a single such gate suffices; for $k > 2$, multiple gates connecting adjacent factors are required. These gates satisfy both density conditions: they restore irreducibility by coupling the local Clifford structures, and they certify infiniteness by lying within projective distance $1/2$ of the identity (Theorem~\ref{thm:emergent_universality}). Accordingly, the needed diagonal phases can be obtained without taking a separate explicit diagonal magic gate as a primitive resource. Again, not every non-Clifford gate suffices (Theorem~\ref{thm:tensor_mixing}).

\end{enumerate}

In particular, permutations already provide universal resources in every dimension $d \ge 4$, although the relevant permutation depends on the arithmetic structure of $d$. Here the threshold $d \ge 4$ refers to the existence of some universal permutation across the trichotomy, whereas the stronger prime-dimensional statement that the transposition $(0\;1)$ is universal requires $p \ge 5$ because for $p=2,3$ every permutation is affine and, hence, Clifford.

The trichotomy highlights three kinds of resource gates: diagonal phase gates generalizing the qubit $T$ gate, simple basis permutations, and arithmetic gates acting between coprime factors of a composite register. Their role in this paper is structural. We identify the conditions under which such gates promote the single-qudit Clifford group to universality, and classify them as either $T$-type phase gates or CNOT-type permutation gates.

This dimension dependence matters because higher-dimensional and hybrid-dimensional qudit platforms are increasingly accessible experimentally~\cite{RingbauerEtAl2022, Low2025, Chicco_2023}, while the standard prime-dimensional Clifford~+~$T$ picture does not directly describe them. 
A local account of how universality changes beyond the prime case is still lacking. The trichotomy established here provides that account: the prime factorization of $d$ determines the obstruction to single-qudit universality and identifies local resources that overcome it in each regime. Once these obstructions are resolved, universality in the multi-qudit setting ensues.

\textbf{Structure of the paper.} Section~\ref{sec:preliminaries} sets up the universality criterion, recalls the Brylinski reduction from local to multi-qudit universality, and introduces the Clifford-group framework used throughout. Section~\ref{sec:geometric} proves the infiniteness criterion based on projective distance to the identity. Sections~\ref{sec:prime}--\ref{sec:composite} treat the three branches of the trichotomy in turn. Section~\ref{sec:discussion} concludes with implications, limitations, and open questions.

\section{Preliminaries}\label{sec:preliminaries}

\subsection{Criterion for Universality}

The practical realization of quantum algorithms relies on the ability to approximate arbitrary unitary transformations using a finite set of elementary gates. In the context of qudits, we distinguish between universality on a single local register and universality on a multi-qudit system. We adopt a topological definition based on the operator norm, which captures the physical requirement that any target unitary can be realized within an arbitrary error tolerance $\epsilon$.

\begin{definition}
A finite set of gates $\mathcal{G} \subset \SU(d)$ is \emph{single-qudit universal} if the group generated by $\mathcal{G}$, denoted $\langle \mathcal{G} \rangle$, is dense in $\SU(d)$ with respect to the operator norm topology. That is, for every $U \in \SU(d)$ and every $\epsilon > 0$, there exists a finite sequence of gates $g_1, \dots, g_k \in \mathcal{G}$ such that $\| U - g_k \cdots g_1 \| < \epsilon$.
\end{definition}

The focus of this work is primarily on establishing single-qudit universality. This focus is justified by the fact that the complexity of the universality problem resides almost entirely in the local structure. Once the relevant single-qudit obstruction has been identified, extending this to universal quantum computation on an array of qudits requires only a single additional gate $V$ that breaks the preservation of the local tensor structure (an \emph{entangling resource}). This reduction is formalized by the following result:

\begin{theorem}[Brylinski and Brylinski~\cite{Brylinski2002}] \label{thm:brylinski}
Let $\mathcal{G}_{\mathrm{local}} \subset \SU(d)$ be a single-qudit universal set and $V \in \SU(d^2)$ be a two-qudit gate. Then, the set $\mathcal{G}_{\mathrm{local}} \cup \{V\}$ is universal for quantum computation on $n$ qudits if and only if $V$ cannot be decomposed into local operations, even up to a permutation of the qudits. Explicitly, universality is achieved if and only if for all $U_1, U_2 \in \SU(d)$:
\begin{equation}\label{eq:brylinski_cond}
    V \neq U_1 \otimes U_2 \quad \text{and} \quad V \neq (U_1 \otimes U_2) \cdot \mathrm{SWAP}.
\end{equation}
\end{theorem}

Consequently, establishing universality reduces to identifying when a finite set of gates generates a dense subgroup of the special unitary group. Proving density directly via topological arguments is often hard. The main difficulty lies in distinguishing whether a finitely generated subgroup is dense or trapped within a proper, infinite closed Lie subgroup of lower dimension. 

Instead, we employ an algebraic criterion based on the rigidity of simple Lie algebras. This criterion transforms the topological problem of density into two algebraic verifications: the irreducibility of the action on the associated Lie algebra and the infiniteness of the group.

To establish the notation, we introduce the necessary infinitesimal structures. The Lie algebra $\mathfrak{su}(d)$ consists of all trace-zero skew-Hermitian matrices:
\begin{equation}
    \mathfrak{su}(d) = \{X \in M_d(\mathbb{C}) : X^\dagger = -X, \, \tr(X) = 0\}.
\end{equation}
This forms a real vector space of dimension $d^2-1$. Its complexification is the special linear Lie algebra $\mathfrak{sl}(d,\mathbb{C}) = \{X \in M_d(\mathbb{C}) : \tr(X) = 0\}$. The \emph{adjoint representation} $\Ad\colon \SU(d) \to \GL(\mathfrak{su}(d))$ is defined by conjugation:
\begin{equation}
    \Ad_g(X) = gXg^\dagger,
\end{equation}
for $g \in \SU(d)$ and $X \in \mathfrak{su}(d)$. Since unitary conjugation preserves the Hilbert-Schmidt inner product, this action is unitary. The irreducibility of this action is insensitive to complexification: the action on $\mathfrak{su}(d)$ is irreducible if and only if the complexified action on $\mathfrak{sl}(d,\mathbb{C})$ is irreducible.

We now establish the connection between the subgroup dimension and the representation.

\begin{lemma}\label{lem:adjoint_irreducibility}
Let $K$ be a closed subgroup of $\SU(d)$ with positive dimension. Then, the adjoint representation of $K$ on $\mathfrak{su}(d)$ is reducible unless $K = \SU(d)$.
\end{lemma}

\begin{proof}
Let $\mathfrak{k} = \mathrm{Lie}(K) \subseteq \mathfrak{su}(d)$ be the Lie algebra of $K$. Since $K$ has positive dimension, $\mathfrak{k} \neq \{0\}$.
The Lie algebra of a subgroup is invariant under the adjoint action of the group. Thus, $\mathfrak{k}$ is a $K$-invariant subspace of $\mathfrak{su}(d)$.
If we assume the adjoint action of $K$ on $\mathfrak{su}(d)$ is irreducible, the only non-zero invariant subspace can be the entire space itself. Therefore, we must have $\mathfrak{k} = \mathfrak{su}(d)$.

Since $\SU(d)$ is connected and $K$ is a closed subgroup possessing the full Lie algebra of the ambient group, the connected component of the identity $K^\circ$ must coincide with $\SU(d)$. Since $K \subseteq \SU(d)$, it follows that $K = \SU(d)$.
\end{proof}

This lemma leads to the following criterion for single-qudit universality.

\begin{theorem}\label{thm:density_criterion}
Let $\mathcal{A} \subset \SU(d)$ be a finite set of gates. Then $\mathcal{A}$ is single-qudit universal if and only if:
\begin{enumerate}
    \item The group $\langle\mathcal{A}\rangle$ is infinite, and
    \item The adjoint representation of $\langle\mathcal{A}\rangle$ on $\mathfrak{su}(d)$ (equivalently on $\mathfrak{sl}(d,\mathbb{C})$) is irreducible.
\end{enumerate}
\end{theorem}

\begin{proof}
Let $K = \overline{\langle \mathcal{A} \rangle}$ be the topological closure of the generated group in $\SU(d)$.

$(\Rightarrow)$ Suppose $K = \SU(d)$. Since $\SU(d)$ has positive dimension, $\langle \mathcal{A} \rangle$ must be infinite. Furthermore, since $\mathfrak{su}(d)$ is a simple Lie algebra (for $d \ge 2$), the adjoint action of $\SU(d)$ is irreducible. Since $\langle \mathcal{A} \rangle$ is dense in $K$, any subspace invariant under $K$ is invariant under $\langle \mathcal{A} \rangle$. Thus, the dense subgroup must also act irreducibly.

$(\Leftarrow)$ Suppose $\langle \mathcal{A} \rangle$ is infinite and acts irreducibly on $\mathfrak{su}(d)$. Since the subgroup is infinite and contained in a compact group, its closure $K$ cannot be a finite set. Being a closed subgroup of a Lie group, $K$ is itself a Lie group with $\dim(K) > 0$.
The irreducibility of $\langle \mathcal{A} \rangle$ implies that the adjoint action of $K$ on $\mathfrak{su}(d)$ is also irreducible. We now apply Lemma~\ref{lem:adjoint_irreducibility} to $K$. Since $K$ has positive dimension and acts irreducibly, we conclude $K = \SU(d)$.
\end{proof}

In practice, we will often verify the second condition using the complex Lie algebra $\mathfrak{sl}(d,\mathbb{C})$, which simplifies calculations by allowing the use of complex bases such as the generalized Pauli operators.

\subsection{The Clifford Group and its Structure}\label{subsec:clifford_structure}

We now establish the structural properties of the Clifford group essential for our universality arguments. Consider a $d$-dimensional Hilbert space $\mathcal{H}_d \cong \mathbb{C}[\mathbb{Z}_d]$ equipped with the computational basis $\{\ket{k} : k \in \mathbb{Z}_d\}$. We distinguish between the index set viewed as the additive abelian group $\mathbb{Z}_d$, and the ring of integers modulo $d$, denoted by $\mathbb{Z}/d\mathbb{Z}$.
This is the standard generalized Pauli/Clifford framework for arbitrary qudit dimension, naturally expressed in modular-arithmetic terms~\cite{CliffordPauliGeneralization}.

The generalized Pauli operators $X$ and $Z$ are defined by their action on the computational basis:
\begin{equation*}
    X\ket{j} = \ket{j+1}, \qquad Z\ket{j} = \omega^j\ket{j},
\end{equation*}
where arithmetic in the indices is taken modulo $d$, and $\omega = e^{2\pi i/d}$ is the primitive $d$-th root of unity. These operators satisfy the  commutation relation $$Z X = \omega X Z.$$
The \emph{Heisenberg group} $H(d)$ (sometimes called the generalized Pauli group or just the Pauli group~\cite{Gheorghiu2014, deBeaudrap2013}) consists of:
\begin{equation*}
    H(d) = \begin{cases}
         \{\omega^l X^a Z^b : a,b,l \in \mathbb{Z}_d\} & \text{if } d \text{ is odd}, \\
         \{\zeta^l X^a Z^b : a,b \in \mathbb{Z}_d, \, l \in \mathbb{Z}_{2d}\} & \text{if } d \text{ is even},
    \end{cases}
\end{equation*}
where $\zeta = e^{\pi i / d}$ is a primitive $2d$-th root of unity. It is convenient to label the Pauli operators by vectors $u = (a,b) \in \mathbb{Z}_d \times \mathbb{Z}_d$, denoting $V_u = X^a Z^b$. The set $\{V_u : u \neq (0,0)\}$ forms an orthogonal basis for the Lie algebra $\mathfrak{sl}(d, \mathbb{C})$ with respect to the Hilbert-Schmidt inner product.

\subsubsection*{Projective vs.\ Finite Clifford Groups}

In the quantum information literature, the term ``Clifford group'' is often used interchangeably to refer to distinct mathematical objects. It is relevant for our  analysis to  distinguish between the \emph{projective group} and its finite unitary representations.

\begin{definition}
The (projective) single-qudit Clifford group, denoted $C(d)$, is defined as the normalizer of the Heisenberg group in the unitary group, modulo global phases:
\begin{equation}
\begin{split}
    C(d) &= N_{\U(d)}(H(d)) / \U(1) \\
         &= \{ U \in \U(d) : U H(d) U^\dagger = H(d) \} / \U(1).
\end{split}
\end{equation}
\end{definition}

The adjoint action of any representative $U \in C(d)$ maps Pauli operators to Pauli operators, inducing an automorphism on the index space. This action corresponds to the special linear group over the ring $\mathbb{Z}/d\mathbb{Z}$. This structural rigidity is characterized by the short exact sequence~\cite{Tolar_2018}:
\begin{equation}\label{eq:clifford_exact_seq}
    1 \longrightarrow \mathbb{Z}_d^2 \longrightarrow C(d) \xrightarrow{\Psi} \SL(2, \mathbb{Z}/d\mathbb{Z}) \longrightarrow 1.
\end{equation}
Here, the kernel $\mathbb{Z}_d^2 \cong H(d)/\U(1)$ corresponds, modulo global phases, to the Pauli operators themselves. Since elements of the Heisenberg group commute up to a scalar, their adjoint action on the Pauli basis is trivial on the index space (mapping $V_u \mapsto \omega^k V_u$ without changing $u$).
In contrast, the quotient $\SL(2, \mathbb{Z}/d\mathbb{Z})$ characterizes the non-trivial structural changes: these gates actively transform the Pauli set by implementing group automorphisms on the indices, $u \mapsto \Psi(u)$.

While $C(d)$ is the natural object for studying symmetries, it is an abstract quotient group. The full normalizer $N_{\U(d)}(H(d))$ within the unitary group is continuous, as it includes the circle group $\U(1)$ of global phases. To properly formulate the universality problem---which asks whether a finite set of gates generates a dense subgroup---we must identify a concrete \emph{finite} subgroup of unitary matrices that surjects onto $C(d)$.

\begin{definition}
We define the \emph{finite single-qudit Clifford group}, denoted $\mathcal{C}_d \subset \U(d)$, as the subgroup generated by the  operator $X$, the generalized Hadamard gate $H_d$, and the phase gate $P$:
\begin{equation*}
    H_d\ket{j} = \frac{1}{\sqrt{d}} \sum_{k=0}^{d-1} \omega^{jk} \ket{k}, \qquad P\ket{j} = \omega^{j(j+\rho_d)/2}\ket{j},
\end{equation*}
where $\rho_d = 0$ if $d$ is even and $\rho_d = 1$ if $d$ is odd.
\end{definition}

The two definitions play different roles in the paper. The projective group $C(d)$ is the intrinsic normalizer of the Heisenberg group modulo global phases, and it is the natural object for structural statements such as exact sequences and maximality modulo scalars. By contrast, the finite group $\mathcal{C}_d \subset \U(d)$ is a concrete gate set generated by standard Clifford unitaries, and it is the object we use when discussing generation by gates and density in $\SU(d)$. The natural projection $\U(d) \to \mathrm{PU}(d)$ maps $\mathcal{C}_d$ onto $C(d)$, so the two viewpoints differ only by central phase factors.

To validate that the finite group $\mathcal{C}_d = \langle X, H_d, P \rangle$ serves as the correct lifting of the projective Clifford group, we analyze the action of its generators under the homomorphism $\Psi$. The images of the Hadamard and phase gates correspond to the standard generators of the special linear group over the ring $\mathbb{Z}/d\mathbb{Z}$:
\begin{equation*}
    \Psi(H_d) = \begin{pmatrix} 0 & -1 \\ 1 & 0 \end{pmatrix}, \qquad 
    \Psi(P) = \begin{pmatrix} 1 & 0 \\ 1 & 1 \end{pmatrix}.
\end{equation*}
Since these matrices generate $\SL(2, \mathbb{Z}/d\mathbb{Z})$, the group $\mathcal{C}_d$ surjects onto the symplectic quotient. Furthermore, the relation $Z = H_d X H_d^\dagger$ implies that $\mathcal{C}_d$ contains the full discrete Heisenberg group $H(d)$ (up to the phases required for group closure). Consequently, $\mathcal{C}_d$ fully captures the exact sequence structure defined in Eq.~\eqref{eq:clifford_exact_seq}.

\begin{remark}
Unless stated otherwise, statements about concrete gate generation in this paper are made using the finite lift $\mathcal{C}_d$. The group $\mathcal{C}_d$ fits into the short exact sequence:
\begin{equation}\label{eq:exact_sequence_finite_clifford}
    1 \longrightarrow H(d) \longrightarrow \mathcal{C}_d \longrightarrow \SL(2, \mathbb{Z}/d\mathbb{Z}) \longrightarrow 1.
\end{equation}
Unlike the full normalizer in $\U(d)$, which contains the continuous circle group $\U(1)$, $\mathcal{C}_d$ is a finite group. Its order is strictly determined by the size of the $\SL(2, \mathbb{Z}/d\mathbb{Z})$ group and the scalar phases intrinsic to $H(d)$:
\begin{equation}
    |\mathcal{C}_d| = |H(d)| \cdot |\SL(2, \mathbb{Z}/d\mathbb{Z})|.
\end{equation}

\end{remark}

The finite Clifford group factorizes over coprime dimensions---a property previously identified by Looi and Griffiths \cite{Looi-Griffiths} for stabilizer states---which allows us to reduce the problem to independent prime-power sectors.

Because the Pauli and Clifford groups depend on the chosen computational basis, the tensor-product description used in the coprime composite case should be understood relative to a fixed arithmetic identification. Throughout this subsection, we keep the computational basis $\{\ket{x}: x \in \mathbb{Z}_d\}$ of $\mathbb{C}^d$ and, for $d=d_1\cdots d_n$ with pairwise coprime factors, use the map
\[
x \mapsto (x \bmod d_1,\dots,x \bmod d_n)
\]
to identify basis states as
\[
\ket{x} \longleftrightarrow \ket{x \bmod d_1}\otimes \cdots \otimes \ket{x \bmod d_n}.
\]
All subsequent statements about the factorization of the Clifford group and about intra-qudit mixing gates are made with respect to this fixed identification. In particular, calling an intra-qudit CNOT a non-Clifford resource is always relative to the same global $d$-dimensional Clifford structure determined by this basis choice.

\begin{proposition}\label{prop:clifford_decomposition}
Let $d = d_1 \cdots d_n$, where the factors are pairwise coprime ($\gcd(d_i, d_j) = 1$ for $i \neq j$). Then there is a group isomorphism:
\begin{equation*}
    \mathcal{C}_d \cong \mathcal{C}_{d_1} \times \cdots \times \mathcal{C}_{d_n}.
\end{equation*}
\end{proposition}

\begin{proof}
It suffices to establish the case $n=2$ with coprime factors $d_1, d_2$.
We define the map $\Phi: \mathcal{C}_{d_1} \times \mathcal{C}_{d_2} \to \U(d_1 d_2)$ via the tensor product $\Phi(U_1, U_2) = U_1 \otimes U_2$. We show this maps generators to generators, inducing an isomorphism onto $\mathcal{C}_d$.

Recall that by the Chinese Remainder Theorem, the map $\chi: \mathbb{Z}/d\mathbb{Z} \to \mathbb{Z}/{d_1}\mathbb{Z} \times \mathbb{Z}/{d_2}\mathbb{Z}$ given by $x \mapsto (x \bmod d_1, x \bmod d_2)$ is a ring isomorphism.

Since the Clifford generators are defined entirely by arithmetic operations on the basis indices, the CRT implies their factorization:
\begin{enumerate}
    \item The shift operator ($X$): The action $X\ket{x} = \ket{x+1}$ relies on addition. Under $\chi$, adding $1$ modulo $d$ corresponds to adding $(1, 1)$ in the product ring. Thus, $X_d \cong X_{d_1} \otimes X_{d_2}$.
    \item Hadamard ($H$): The Fourier transform kernel $\omega^{xy}$ factorizes because the inner product in $\mathbb{Z}_d$ decomposes into the sum of local inner products (modulo integer scaling), yielding $H_d \cong H_{d_1} \otimes H_{d_2}$.
    \item Phase Gate ($P$): The gate applies a phase $\omega^{Q(x)}$, where $Q(x) = x(x+\rho_d)/2$ is a quadratic polynomial. Since $\chi$ is a ring homomorphism, it preserves polynomial evaluation: computing $Q(x)$ modulo $d$ is equivalent to computing pairs $(Q(x_1) \bmod d_1, Q(x_2) \bmod d_2)$. Consequently, the diagonal operator factorizes as $P_d \cong P_{d_1} \otimes P_{d_2}$.
\end{enumerate}
Since the generators map bijectively and the groups have the same order, $\Phi$ is an isomorphism.
\end{proof}

\subsection{Permutations in the Clifford Group}\label{subsec:clifford_permutations}

While the Clifford group acts mainly by creating superpositions via the Hadamard gate, it also contains a subgroup of classical reversible operations corresponding to permutations of the computational basis. Characterizing this subgroup is essential for identifying which classical gates constitute non-Clifford resources.

Let $S_d$ be the symmetric group of degree $d$. For any permutation $\pi \in S_d$, we define the associated permutation matrix $P_\pi \in \U(d)$ by its action on the basis states: $P_\pi |x\rangle = |\pi(x)\rangle$.

\begin{proposition}\label{prop:clifford_permutations}
A permutation matrix $P_\pi$ belongs to the Clifford group $\mathcal{C}_d$ if and only if the underlying permutation $\pi: \mathbb{Z}_d \to \mathbb{Z}_d$ is an invertible affine transformation over the ring $\mathbb{Z}_d$. That is, $\pi$ is of the form:
\begin{equation}
    \pi(x) = a x + b \pmod d,
\end{equation}
where $a \in \mathbb{Z}_d^\times$ (a unit in $\mathbb{Z}_d$) and $b \in \mathbb{Z}_d$. Consequently, the subgroup of Clifford permutations is isomorphic to the affine group $\mathrm{Aff}(1, \mathbb{Z}_d)$.
\end{proposition}

\begin{proof}
$(\Rightarrow)$ Assume $P_\pi \in \mathcal{C}_d$. By definition, conjugation by $P_\pi$ must map the Heisenberg group $H(d)$ to itself.

Consider the action on the Pauli-$Z$ operator, which is diagonal with entries $\omega^x$ where $\omega = e^{2\pi i/d}$. Conjugation permutes the diagonal entries:
\begin{equation*}
    P_\pi Z P_\pi^\dagger = \sum_{x \in \mathbb{Z}_d} \omega^x |\pi(x)\rangle\langle \pi(x)| = \sum_{y \in \mathbb{Z}_d} \omega^{\pi^{-1}(y)} |y\rangle\langle y|.
\end{equation*}
For the resulting diagonal operator to belong to $H(d)$, it must be of the form $\omega^{c_0} Z^k$ for some $k, c_0 \in \mathbb{Z}_d$. Comparing the diagonal entries, we require $\omega^{\pi^{-1}(y)} = \omega^{c_0 + ky}$ for all $y \in \mathbb{Z}_d$. This implies:
\begin{equation*}
    \pi^{-1}(y) \equiv k y + c_0 \pmod d.
\end{equation*}
Since $\pi^{-1}$ is a bijection, the coefficient $k$ must be a unit in $\mathbb{Z}_d$; otherwise, the map $y \mapsto ky + c_0$ would fail to be injective. Thus, $\pi^{-1}$ is an affine map with invertible linear part, and consequently $\pi$ itself must be affine: $\pi(x) = ax + b$ with $a \in \mathbb{Z}_d^\times$.

$(\Leftarrow)$ Conversely, assume $\pi(x) = ax + b$ with $a \in \mathbb{Z}_d^\times$. We verify the conjugation action on the generators $Z$ and $X$ of the Heisenberg group.

For the Pauli-$Z$ operator, a direct calculation using $\pi^{-1}(y) = a^{-1}(y - b)$ gives:
\begin{equation*}
    P_\pi Z P_\pi^\dagger = \sum_{y \in \mathbb{Z}_d} \omega^{a^{-1}(y-b)} |y\rangle\langle y| = \omega^{-a^{-1}b} Z^{a^{-1}},
\end{equation*}
which belongs to $H(d)$.

For the  operator $X$, we compute:
\begin{align*}
    P_\pi X P_\pi^\dagger |y\rangle &= P_\pi X |\pi^{-1}(y)\rangle = P_\pi | \pi^{-1}(y) + 1 \rangle \\
    &= | \pi( \pi^{-1}(y) + 1 ) \rangle.
\end{align*}
Substituting $\pi^{-1}(y) = a^{-1}(y-b)$, the argument of $\pi$ becomes:
\begin{equation*}
    \pi\big( a^{-1}(y-b) + 1 \big) = a \cdot a^{-1}(y-b) + a + b = y + a.
\end{equation*}
Thus, $P_\pi X P_\pi^\dagger = X^a$. Since $a$ is a unit, $X^a$ generates the same cyclic subgroup as $X$ and belongs to $H(d)$.

Since $P_\pi$ maps the generators of $H(d)$ to elements of $H(d)$, we conclude $P_\pi \in \mathcal{C}_d$.
\end{proof}

\begin{remark}\label{rem:non_affine_permutations}
The order of the affine group is $|\mathrm{Aff}(1, \mathbb{Z}_d)| = \varphi(d) \cdot d$, where $\varphi$ denotes Euler's totient function.

For $d = 2$, we have $|S_2| = 2$ and $|\mathrm{Aff}(1, \mathbb{Z}_2)| = 1 \cdot 2 = 2$. For $d = 3$, we have $|S_3| = 6$ and $|\mathrm{Aff}(1, \mathbb{Z}_3)| = 2 \cdot 3 = 6$. In both cases, the groups coincide: every permutation is affine and thus a Clifford operation.

However, for $d \ge 4$, the inequality $\varphi(d) \cdot d < d!$ implies the strict inclusion $\mathrm{Aff}(1, \mathbb{Z}_d) \subsetneq S_d$. Consequently, non-affine permutations exist. For instance, in $d = 4$, the transposition $(0 \; 1)$---which swaps $|0\rangle$ and $|1\rangle$ while fixing $|2\rangle$ and $|3\rangle$---is a non-Clifford gate despite being a classical reversible operation.
\end{remark}

\subsection{Diagonal Gates and the \texorpdfstring{$T_s$}{T\_s}-gate Family}\label{subsec:general_Ts}

To characterize diagonal non-Clifford operators, for any dimension $d$, we introduce a generic family of diagonal gates.

\begin{definition}
Let $\zeta: \mathbb{Z}_d \to \U(1)$ be a phase function. We define the associated diagonal gate $T_\zeta$ by its action on the computational basis:
\begin{equation}
    T_\zeta |x\rangle = \zeta(x) |x\rangle, \quad \text{for } x \in \mathbb{Z}_d.
\end{equation}
\end{definition}

To determine whether a specific $T_\zeta$ belongs to the Clifford group, we examine its conjugation action on the operator $X$. A simple calculation shows that $T_\zeta X T_\zeta^\dagger = D X$, where $D$ is a diagonal operator encoding the relative phases. This motivates the introduction of the \emph{differential} of the phase function.

\begin{definition}
The differential (or 1-coboundary) of a function $\zeta: \mathbb{Z}_d \to \U(1)$ is the map $\delta_\zeta: \mathbb{Z}_d \times \mathbb{Z}_d \to \U(1)$ defined by:
\begin{equation}\label{eq:1-coboundary}
    \delta_\zeta(x, y) = \frac{\zeta(x+y)}{\zeta(x)\zeta(y)}.
\end{equation}
\end{definition}

Recall that a function $b: \mathbb{Z}_d \times \mathbb{Z}_d \to \U(1)$ is called a \emph{bicharacter} if it is linear in both arguments.

\begin{lemma}\label{lem:clifford_diagonal}
Let $d$ be any integer. A diagonal gate $T_\zeta$ belongs to the Clifford group $\mathcal{C}_d$ if and only if its differential $\delta_\zeta$ is a bicharacter.
\end{lemma}

\begin{proof}
The conjugation action on the shift operator $X^y$ is given by:
\[
    T_\zeta X^y T_\zeta^\dagger |k\rangle = \frac{\zeta(k)}{\zeta(k-y)} X^y |k\rangle.
\]
For the result to be a generalized Pauli operator (of the form $Z^a X^y$), the coefficient function must be a character (linear in $k$). In terms of the differential, this condition requires $\delta_\zeta(k, -y)$ to be linear in $k$. Since this must hold for all $y$, $\delta_\zeta$ must be bilinear.

Note that for the cyclic group $\mathbb{Z}_d$, any symmetric bicharacter takes the simple form $b(x, y) = \omega^{c xy}$, where $\omega = e^{2\pi i / d}$ is the primitive $d$-th root of unity and $c$ is an integer constant. Thus, the condition implies that $\zeta(x)$ must be a quadratic phase polynomial $\omega^{ax^2+bx}$.
\end{proof}

We specifically focus on the following $T_s$ family.

\begin{definition}[The $T_s$-gate Family]
For any dimension $d \ge 2$ and any integer resolution $s \ge 2$, we define the diagonal gate $T_s$ acting on the computational basis as:
\begin{equation}\label{eq:T_s}
    T_s = \sum_{k \in \mathbb{Z}_d} \exp\left( \frac{2\pi i k}{s} \right) |k\rangle\langle k|.
\end{equation}
\end{definition}

For instance, $T_d$ recovers the Pauli $Z$ gate (order $d$), and $T_{2d}$ corresponds to the generalized phase gate. To determine the ``Cliffordness'' of $T_s$, we analyze its conjugation action on $X$:
\begin{equation}
    \Ad_{T_s}(X) = T_s X T_s^\dagger = \exp\left( \frac{2\pi i}{s} \right) X.
\end{equation}
For $T_s$ to belong to the Clifford group $\mathcal{C}_d$, it must normalize the Heisenberg group. Since the result is proportional to $X$, this condition holds if and only if the phase $e^{2\pi i / s}I$ is in the Heisenberg group:
\begin{itemize}
    \item Odd Dimensions: The phases are generated by $\omega = e^{2\pi i / d}$. Thus, $T_s \in \mathcal{C}_d \iff s \mid d$.
    \item Even Dimensions: The Clifford group generates phases of order $2d$. Thus, $T_s \in \mathcal{C}_d \iff s \mid 2d$.
\end{itemize}

This leads to a unified criterion for when $T_s$ constitutes a non-Clifford resource.

\begin{proposition}\label{prop:Ts_resource}
The gate $T_s$ is a non-Clifford resource for $\SU(d)$ if and only if the order $s$ satisfies:
\begin{equation}
    s \nmid K_d, \quad \text{where } K_d = \begin{cases} d & \text{if } d \text{ is odd}, \\ 2d & \text{if } d \text{ is even}. \end{cases}
\end{equation}
\end{proposition}
\section{A Criterion for Infiniteness}\label{sec:geometric}

Proving that a finitely generated matrix group is infinite is, in general, a non-trivial algebraic problem~\cite{detinko2013deciding}. While the irreducibility of the adjoint action ensures that the group is not trapped in a proper Lie subgroup of lower dimension, it does not rule out the possibility that the group is finite. To address this, we introduce a topological condition based on the discrete nature of finite groups. 

Since we are interested in subgroups of $\SU(d)$, but physical operations are equivalent up to global phases, the appropriate topological setting is the projective unitary group $\mathrm{PU}(d) = \SU(d)/Z(\SU(d)) \cong \U(d)/\U(1)$. To capture the proximity to the identity in this quotient space, we employ a \emph{projective operator distance}.

\begin{definition}
The projective distance of $U \in \SU(d)$ is defined as:
\begin{equation}\label{eq:proj_metric_def}
    \mathrm{dist}_{\mathrm{proj}}(U, I) \coloneqq \min_{\phi \in [0, 1)} \| e^{2\pi i\phi} U - I \|,
\end{equation}
where $\| A \| \coloneqq \sup_{\ket{\psi} : \braket{\psi}{\psi}=1} \| A \ket{\psi} \|$ denotes the standard operator norm.
\end{definition}

Note that the minimization is well-defined because the function $f(\phi) = \| e^{2\pi i\phi} U - I \|$ is continuous on the compact interval $[0, 1]$. Also, while $U \in \SU(d)$, the optimal representative $e^{2\pi i\phi}U$ may lie in $\U(d)$.

\begin{theorem}\label{thm:infinite_criterion}
Let $G \subseteq \SU(d)$ be a subgroup that acts irreducibly on the Lie algebra $\mathfrak{su}(d)$ via the adjoint representation. Then, $G$ is infinite if and only if there exists an element $h \in G$ such that:
\begin{equation}
    0 < \mathrm{dist}_{\mathrm{proj}}(h, I) < \frac{1}{2}.
\end{equation}
\end{theorem}

The lower bound ($>0$) guarantees that $h$ is not a scalar, while the upper bound ($<1/2$) places it within a neighborhood of the identity prohibited for finite primitive groups.

\subsection{Proof of the Criterion}

To prove Theorem~\ref{thm:infinite_criterion}, we first establish a lemma connecting the irreducibility of the adjoint representation to the structure of normal abelian subgroups. Recall that a linear group $G \subseteq \GL(V)$ is called \emph{imprimitive} if it preserves a decomposition of the space $V = \bigoplus V_i$ into subspaces permuted by $G$. If no such decomposition exists and $G$ is irreducible, it is called \emph{primitive}.

\begin{lemma}\label{lem:normal_abelian_is_multiple_identity}
Let $G \subseteq \SU(d)$ be a subgroup. If the adjoint representation of $G$ on $\mathfrak{su}(d)$ is irreducible, then $G$ acts primitively on $\mathbb{C}^d$, and every normal abelian subgroup of $G$ is contained in the center $Z(\SU(d)) = \{\omega I : \omega^d = 1\}$.
\end{lemma}

\begin{proof}
First, we show that irreducibility of the adjoint action implies primitivity on $\mathbb{C}^d$. Suppose $G$ is not primitive. Then it acts either reducibly or imprimitively on $\mathbb{C}^d$:
\begin{itemize}
    \item If $G$ acts reducibly, there exists a proper invariant subspace $W \subset \mathbb{C}^d$. The subalgebra $\mathfrak{k} = \{X \in \mathfrak{su}(d) : XW \subseteq W\}$ would then be a proper non-zero invariant subspace for the adjoint representation.
    \item If $G$ is imprimitive, preserving a decomposition $V = \bigoplus_{i=1}^m V_i$, then the subalgebra $\mathfrak{h} = \{X \in \mathfrak{su}(d) : X V_i \subseteq V_i \, \forall i\}$ is a non-zero subspace invariant under conjugation by $G$.
\end{itemize}
In both cases, the adjoint representation would be reducible, contradicting the hypothesis. Thus, $G$ is primitive.

Now, let $A \trianglelefteq G$ be a normal abelian subgroup. Since elements of $A$ commute, they share a simultaneous eigenspace decomposition $V = \bigoplus_{\chi} V_\chi$. Since $A$ is normal, $G$ permutes these eigenspaces. If there were more than one distinct eigenspace, this would constitute a system of imprimitivity. Since $G$ is primitive, there must be only one component $V = V_\chi$. This implies that every $a \in A$ acts as a scalar $a = \chi(a)I$, so $A \subseteq Z(SU(d))$.
\end{proof}

We can now prove the main theorem.

\begin{proof}[Proof of Theorem~\ref{thm:infinite_criterion}]
$(\Rightarrow)$ Suppose $G$ is infinite. Since the adjoint representation is irreducible, Theorem~\ref{thm:density_criterion} implies that $G$ is dense in $\SU(d)$ with respect to the  norm topology.

We define the function $F: \SU(d) \to \mathbb{R}^{\geq 0}$ by $F(U) = \mathrm{dist}_{\mathrm{proj}}(U, I)$. This function is continuous, as it is the composition of continuous operations. Consider the set:
\[ \mathcal{U} = \{ U \in \SU(d) : F(U) < 1/2 \} = F^{-1}\big( [0, 1/2) \big). \]
Since $F$ is continuous, $\mathcal{U}$ is an open set in $\SU(d)$. Furthermore, $\mathcal{U}$ is non-empty as it contains the identity ($F(I)=0$).

A dense subgroup $G$ must intersect any non-empty open set infinitely many times. Therefore, $G \cap \mathcal{U}$ contains infinitely many elements. We can thus choose an element $h \in G \cap \mathcal{U}$ that is not a scalar multiple of the identity (i.e., $F(h) > 0$). This $h$ satisfies the condition $0 < \mathrm{dist}_{\mathrm{proj}}(h, I) < 1/2$.

$(\Leftarrow)$ We proceed by contradiction. Suppose $G$ is finite and contains $h$ with $0 < \mathrm{dist}_{\mathrm{proj}}(h, I) < 1/2$.
Let $\phi$ be the phase achieving the minimum distance, and define $h' = e^{i\phi}h$.

Consider the extended group $G' = \langle G, e^{i\phi}I \rangle$. We claim that $G'$ is finite.
Since $G$ is finite, $h$ has finite order, implying its eigenvalues are roots of unity. The minimal arc containing the spectrum is delimited by two such eigenvalues. Consequently, the centering phase $\phi$, which corresponds to the midpoint of this arc, is necessarily a rational multiple of $2\pi$. Thus, the scalar generator $e^{i\phi}I$ has finite order. Since $G$ is finite and $e^{i\phi}I$ is a central element of finite order, the generated group $G'$ is finite.

We now rely on a structural property established in the proof of Theorem 5.7 in Ref.~\onlinecite{Dixon1971}: for any finite subgroup $G' \subset \GL(d, \mathbb{C})$, the subgroup $A = \langle \{g \in G' : \|g - I\| < 1/2\} \rangle$ is a normal abelian subgroup.
Since $h' \in G'$ satisfies $\|h' - I\| < 1/2$, it implies $h' \in A$, so $A$ is non-trivial. However, since the adjoint representation of $G'$ (which coincides with that of $G$) is irreducible, Lemma~\ref{lem:normal_abelian_is_multiple_identity} forces any normal abelian subgroup to be central.
Thus, $h'$ must be a scalar matrix. This implies $h$ is a scalar, which contradicts the assumption that $\mathrm{dist}_{\mathrm{proj}}(h, I) > 0$. Therefore, $G$ cannot be finite.
\end{proof}

\subsection{Spectral Characterization via Projective Distance}

In practice, we verify the criterion by inspecting the eigenvalues. To adapt our infiniteness criterion to the projective setting, we identify the group of physical phases $U(1)$ with the quotient group $\mathbb{R}/\mathbb{Z}$ via the exponential map $\theta \mapsto e^{2\pi i \theta}$.

To evaluate the projective distance efficiently, we first introduce a  measure of the angular concentration of the eigenvalues.

\begin{definition}\label{def:spectral_span}
Let $h \in \U(d)$ be a unitary matrix. We define the \emph{spectral span} $\ell(h)$ as  the minimum $\Delta \in [0, 1]$ such that there exists  $\phi_0 \in \mathbb{R}$ and  $\delta_j \in [-\Delta/2, \Delta/2]$ satisfying:
\begin{equation}
    \sigma(h) = \left\{ e^{2\pi i (\phi_0 + \delta_j)} \right\}_{j=1}^d.
\end{equation}
\end{definition}

This quantity $\ell(h)$ represents the ``angular width'' of the spectrum on the unit circle. The following lemma establishes that the projective distance is determined uniquely by $\ell(h)$.

\begin{lemma}\label{lem:spectral_arc}
For any unitary matrix $h \in \U(d)$, the projective distance to the identity is given by:
\begin{equation}
    \mathrm{dist}_{\mathrm{proj}}(h, I) = 2 \sin\left( \frac{\pi \ell(h)}{2} \right).
\end{equation}
\end{lemma}

\begin{proof}
Recall that $\mathrm{dist}_{\mathrm{proj}}(h, I) = \min_{\phi} \max_{j} |e^{2\pi i (\theta_j - \phi)} - 1|$. Using the chordal formula $|e^{i\alpha}-1| = 2|\sin(\alpha/2)|$, this is equivalent to minimizing the maximum angular deviation of the rotated eigenvalues from the identity ($0$).

Let $\ell(h)$ be the spectral span as per Definition~\ref{def:spectral_span}. By existence of the minimum, we can choose the optimal parameters $\phi_0$ and $\{\delta_j\}$ such that the eigenvalues are $\lambda_j = e^{2\pi i (\phi_0 + \delta_j)}$ with maximal deviation $\max |\delta_j| = \ell(h)/2$.
We choose the global phase shift $\phi = \phi_0$ (rotation by $-\phi_0$). This centers the spectrum around $0$, so the angular phases become exactly $2\pi \delta_j$.
The operator norm is determined by the eigenvalue farthest from $1$:
\[
    \| e^{-2\pi i \phi_0} h - I \| = \max_j 2 \left| \sin\left( \frac{2\pi \delta_j}{2} \right) \right| = 2 \sin\left( \pi \frac{\ell(h)}{2} \right).
\]
Any other interval covering the spectrum has length $\ell' \ge \ell(h)$, implying a maximum deviation of at least $\ell'/2$, which would yield a larger distance. Thus, the minimum is attained at the spectral span.
\end{proof}

\begin{corollary}\label{cor:spectral_check}
A subgroup $G \subseteq \SU(d)$ with irreducible adjoint action is infinite if it contains a gate $h$ with spectral span $\ell(h)$ satisfying:
\begin{equation}
    0 < \pi \ell(h) < 2 \arcsin(1/4) \approx 0.5054.
\end{equation}
\end{corollary}


\section{Case I: Prime Dimensions}\label{sec:prime}

We begin with dimensions $d=p$ where $p$ is prime. In this case, the Clifford group has the two properties relevant for our criterion: irreducibility of the adjoint action and maximality modulo scalars. We therefore first establish the correspondence between the group's action on the Lie algebra and the symplectic geometry of the discrete phase space. The same framework will also be used to analyze the obstructions that arise in prime-power dimensions (Sec.~\ref{sec:prime_power}).

\subsection{The Adjoint Action and \texorpdfstring{$\SL(2, \mathbb{Z}/d\mathbb{Z})$}{SL(2,Z/dZ} Orbits}

To explicitly decompose the adjoint representation of the Clifford group $C(d)$ on the Lie algebra $\mathfrak{sl}(d, \mathbb{C})$ into a direct sum of irreducible subrepresentations, we employ the theory of representations of group extensions developed by Clifford~\cite{clifford1937, Serre1977}. 

The definition of the Clifford group implies that its adjoint action is \emph{monomial} with respect to the generalized Pauli basis: elements of the group permute the basis vectors $\{V_u\}$ while applying phase factors. This structure naturally invites the application of Clifford's theory to the short exact sequence characterizing the group (Eq.~\ref{eq:clifford_exact_seq}).

The Pauli operators $\{V_u : u \in \mathbb{Z}_d^2 \setminus \{0\}\}$ form an orthogonal basis for the Lie algebra $\mathfrak{sl}(d, \mathbb{C})$ with respect to the Hilbert-Schmidt inner product. We analyze the adjoint representation $\Ad: C(d) \to \GL(\mathfrak{sl}(d, \mathbb{C}))$. According to Clifford's theory, the restriction of the representation to the normal subgroup $\mathbb{Z}_d^2$ decomposes the space into a direct sum of isotypic components associated with the orbits of the quotient group acting on the character group $\widehat{\mathbb{Z}_d^2}$.

In our case, the Pauli basis elements $V_u$ form a simultaneous eigenbasis for the action of $\mathbb{Z}_d^2$. Explicitly, for any element $u \in \mathbb{Z}_d^2$, which corresponds to a class of Pauli operators $[V_v]$, the adjoint action is given by the commutation relation:
\begin{equation*}
    \Ad_{V_v}(V_u) = V_v V_u V_v^\dagger = \omega^{\langle v, u \rangle} V_u.
\end{equation*}
Thus, each 1-dimensional subspace spanned by $V_u$ corresponds to a distinct character $\chi_u \in \widehat{\mathbb{Z}_d^2}$ defined by $\chi_u(v) = \omega^{\langle v, u \rangle}$. The action of the quotient group $\SL(2, \mathbb{Z}/d\mathbb{Z})$ permutes these characters (and thus the basis vectors) according to the linear map on the indices $u \mapsto \Psi_g(u)$.

Consequently, the adjoint representation decomposes into a direct sum of irreducible $C(d)$-invariant subspaces determined by the $\SL(2, \mathbb{Z}/d\mathbb{Z})$ orbits:
\begin{equation*}
    \mathfrak{sl}(d, \mathbb{C}) = \bigoplus_{\mathcal{O} \in \mathbb{Z}_d^2 / \SL(2, \mathbb{Z}_d)} W_{\mathcal{O}}, \quad \text{where} \quad W_{\mathcal{O}} = \bigoplus_{u \in \mathcal{O}} \mathbb{C} V_u.
\end{equation*}

This leads to the following correspondence:

\begin{lemma}\label{lem:orbit_correspondence}
Let $d$ be any integer. The decomposition of the adjoint representation of $C(d)$ on $\mathfrak{sl}(d, \mathbb{C})$ into irreducible subspaces is in one-to-one correspondence with the orbits of the natural action of the group $\SL(2, \mathbb{Z}/d\mathbb{Z})$ on the non-zero index vectors $\mathbb{Z}_d^2 \setminus \{0\}$. Specifically, for each orbit $\mathcal{O}$, the subspace $W_{\mathcal{O}}$ is an irreducible $C(d)$-module.
\end{lemma}

To apply this correspondence, we require an explicit classification of the $\SL(2, \mathbb{Z}/d\mathbb{Z})$ orbits. For dimensions of the form $d=p^m$ (which includes the prime case $m=1$), the orbits are fully determined by the divisibility of the index vector.

\begin{lemma}\label{lem:orbit_classification}
Let $d=p^m$ be a prime power. Under the action of $\SL(2, \mathbb{Z}/d\mathbb{Z})$ on the  $\mathbb{Z}_d^2$, the quantity $g(u) = \gcd(u_1, u_2, d)$ is an invariant. Two vectors $u, v \in \mathbb{Z}_d^2$ belong to the same orbit if and only if they share this invariant:
\begin{equation*}
\mathcal{O}_u = \mathcal{O}_v \iff \gcd(u_1, u_2, d) = \gcd(v_1, v_2, d).
\end{equation*}
\end{lemma}

\begin{proof}
The greatest common divisor is preserved under integer linear combinations with unit determinant. If $v = M u$ for some $M \in \SL(2, \mathbb{Z}_d)$, then the components of $v$ are linear combinations of the components of $u$. Thus, any common divisor of $u_1, u_2$ (and $d$) must also divide $v_1, v_2$. Since $M$ is invertible, the reverse implication holds, ensuring $g(u) = g(v)$.

Now, we show that any vector $u$ with invariant $g(u) = p^j$ (for $0 \le j \le m$) lies in the orbit of the canonical vector $e_j = (p^j, 0)^T$.

Write $u = (p^j a, p^j b)^T$. Since the total gcd with the modulus $d=p^m$ is $p^j$, the remaining factors must satisfy $\gcd(a, b, p) = 1$. By Bézout's identity, there exist integers $x, y$ such that $ax + by \equiv 1 \pmod{p^m}$. We construct the matrix:
\[
M = \begin{pmatrix} x & y \\ -b & a \end{pmatrix}.
\]
Observe that $\det(M) = ax + by \equiv 1 \pmod{p^m}$, so $M \in \SL(2, \mathbb{Z}/d\mathbb{Z})$. Applying this transformation to $u$:
\[
M \begin{pmatrix} p^j a \\ p^j b \end{pmatrix} = p^j \begin{pmatrix} x & y \\ -b & a \end{pmatrix} \begin{pmatrix} a \\ b \end{pmatrix} = p^j \begin{pmatrix} ax + by \\ -ba + ab \end{pmatrix} = \begin{pmatrix} p^j \\ 0 \end{pmatrix}.
\]
Since any vector with invariant $p^j$ can be mapped to the same canonical state $(p^j, 0)^T$, all such vectors belong to the same orbit.
\end{proof}
\subsection{Irreducibility of the Adjoint Action}

The first condition of the density criterion (Theorem~\ref{thm:density_criterion}) is the irreducibility of the action on the Lie algebra $\mathfrak{sl}(d, \mathbb{C})$.  As formalized by Graydon \textit{et al}.~\cite{Graydon2021}, a finite group forms a unitary 2-design if and only if its adjoint representation on the space of operators decomposes into exactly two irreducible components: the identity and the traceless subspace.

While Graydon \textit{et al}.\ demonstrated that this condition fails for composite dimensions (implying the Clifford group is not a 2-design in those cases), the following theorem establishes the positive counterpart: for prime dimensions, the adjoint action is strictly irreducible.

\begin{theorem}\label{thm:prime_irreducibility}
The adjoint representation of the Clifford group $C(d)$ on $\mathfrak{sl}(d, \mathbb{C})$ is irreducible if and only if $d$ is prime.
\end{theorem}

\begin{proof}
If $d=p$ is prime, the only divisors of $d$ are $1$ and $p$. For any non-zero vector $u \in \mathbb{Z}_p^2 \setminus \{0\}$, the greatest common divisor is $\gcd(u, p) = 1$.
By Lemma~\ref{lem:orbit_classification}, all such vectors belong to a single $\SL(2, \mathbb{Z}/d\mathbb{Z})$ orbit $\mathcal{O}_1$.
Consequently, by Lemma~\ref{lem:orbit_correspondence}, the representation space $\mathfrak{sl}(p, \mathbb{C})$ corresponds to a single irreducible module $W_{\mathcal{O}_1}$. Thus, the action is irreducible.

Conversely, if $d$ is composite, there exists a proper divisor $k$ ($1 < k < d$). The set of vectors with $\gcd(u, d)=k$ forms an orbit disjoint from the set of vectors with $\gcd(u, d)=1$. This implies the existence of at least one proper nonzero invariant subspace, rendering the representation reducible.
\end{proof}

\subsection{Maximality and Universality}

To establish universality for prime dimensions $d=p$, we rely on the fact that any extension of the Clifford group by a non-Clifford gate forces the group to become infinite. This property stems from the maximality of the Clifford group within the finite subgroups of the projective unitary group. While this fact is well-known in quantum information~\cite{Gottesman1999}, a rigorous mathematical proof relies on results from finite group theory.

For the general case of systems with $n$-qudits ($d$ prime), Nebe, Rains, and Sloane~\cite{Nebe2001} established maximality using number-theoretic techniques. Specifically, they utilized $l$-adic integrality and reduction modulo suitable rational primes to apply the classification of subgroups of finite general linear groups. However, as noted in~Ref.~\cite{Nebe2001}, for the single-qudit case ($d=p$) considered here, the sufficiency of the primality condition follows from earlier work by Lindsey~\cite{Lindsey1970}. Unlike the general multi-qudit case, Lindsey's proof relies primarily on character theory and the classification of linear groups of prime degree. Complementing this classical result, we demonstrate that primality is also a \emph{necessary} condition, thereby obtaining the corresponding if-and-only-if statement. For completeness, we provide a self-contained proof of the forward direction based on a Corollary in~\cite{Lindsey1970}.

\begin{theorem}\label{thm:maximality_iff}
The single-qudit Clifford group $\mathcal{C}_d$ is a maximal finite subgroup of the unitary group modulo scalars (i.e., any finite group $G \supsetneq \mathcal{C}_d$ satisfies $G = \langle \mathcal{C}_d, \lambda I \rangle$) if and only if $d$ is a prime number.
\end{theorem}

\begin{proof}
The proof addresses the sufficient and necessary conditions separately.

\noindent \textbf{1. Sufficiency.}
Let $p > 5$ be prime. We will show that $\mathcal{C}_p$ is a maximal finite subgroup of $\mathrm{SU}(p)/U(1)$.

Recall the exact sequence
\[
1 \longrightarrow H(d) \longrightarrow \mathcal{C}_d \longrightarrow \SL(2, \mathbb{Z}/d\mathbb{Z}) \longrightarrow 1,
\]
where $H(p)$ is the Heisenberg group (an extraspecial $p$-group of order $p^3$ and exponent $p$). The group $H(p)$ acts irreducibly on $\mathbb{C}^p$ and is a normal subgroup of $\mathcal{C}_p$.

Suppose, by contradiction, that there exists a finite group $G$ with 
$\mathcal{C}(p) \subsetneq G \subseteq \mathrm{U}(p)/\mathrm{U}(1)$, where the inclusion is strict and then $G$ is not obtained from $\mathcal{C}_p$ by merely adjoining scalar matrices.

The natural $p$-dimensional representation of $G$ is faithful, primitive, and irreducible (primitivity follows since $p$ prime  and then the adjoint representation is irreducible).

We apply the following result of Lindsey:

\begin{theorem}[{\cite[Corollary, p.~58]{Lindsey1970}}]\label{them:Lindsey1970}
Let $G$ be a finite group with a faithful, primitive, irreducible representation of prime degree $p > 5$. If $p^4 \mid |G|$, then $G$ has a normal nonabelian $p$-subgroup $D$ of order $p^3$, and $G/D$ is isomorphic to a subgroup of $\mathrm{SL}(2, \mathbb{Z}/p\mathbb{Z})$.
\end{theorem}

Since $\mathcal{C}_p \subset G$ and $p^4 \mid |\mathcal{C}_p|$, we have $p^4 \mid |G|$. By Theorem~\ref{them:Lindsey1970}, $G$ contains a normal nonabelian $p$-subgroup $D$ of order $p^3$ with $G/D \hookrightarrow \mathrm{SL}(2,\mathbb{Z}/p\mathbb{Z})$.

The group $D$ must coincide with $H(p)$: both are extraspecial $p$-groups of order $p^3$ acting irreducibly on $\mathbb{C}^p$, and such a group is unique up to conjugacy in $\mathrm{GL}(p, \mathbb{C})$.

Therefore, $G$ normalizes $H(p)$, which means $G$ is contained in the normalizer $N_{\mathrm{U}(p)}(H(p))$. But this normalizer, modulo scalars, is precisely the Clifford group $\mathcal{C}_p$. Hence $G = \mathcal{C}(p)$, contradicting our assumption.

For the remaining small primes $p \in \{2, 3, 5\}$, maximality follows from the explicit classification of finite primitive linear groups of degree $p$. If a finite extension $G \supsetneq \mathcal{C}_p$ existed, it would necessarily possess a Sylow $p$-subgroup of order at least $|\mathcal{C}_p|_p \ge p^4$ (accounting for the Heisenberg group and $SL(2,\mathbb{Z}/p\mathbb{Z})$ part). However, the classifications by Blichfeldt~\cite{Blichfeldt1917} and Brauer~\cite{Brauer1967} reveal that no primitive group---other than the normalizers of extraspecial groups---supports such a structure. For instance, for $p=5$, Brauer established that the only other primitive candidates involve central extensions of $A_5, A_6,$ or $\text{PSL}(2,11)$, whose Sylow $5$-subgroups have orders at most $25$, strictly less than the required order. Similar order constraints on the Sylow $p$-subgroups rule out extensions for $p=2$ and $p=3$.

\noindent \textbf{2. Necessity ($d$ composite).}
We prove by contradiction: if $d$ is composite, $\mathcal{C}_d$ is \emph{not} maximal. We distinguish two cases based on the prime factorization of $d$.

\textit{Case A: Prime Powers ($d = p^m$ with $m \geq 2$).}
We explicitly construct a strictly larger finite group by identifying a system of imprimitivity.
Recall that the Pauli operators $X, Z$ generate the Heisenberg group. Consider the subgroup generated by the $p^{m-1}$-th powers of these generators:
\[
\mathcal{A} = \langle X^{p^{m-1}}, Z^{p^{m-1}} \rangle.
\]
We verify the commutativity of these elements. The group commutator is determined by the phase $\omega = e^{2\pi i / p^m}$. Since $[X^{p^{m-1}}, Z^{p^{m-1}}] = \omega^{p^{2m-2}} I$, and $2m-2 \ge m$ for $m \ge 2$, the commutator is the identity. Thus, $\mathcal{A}$ is a normal abelian subgroup.

The Hilbert space decomposes into the direct sum of the simultaneous eigenspaces of $\mathcal{A}$:
\[
\mathbb{C}^d = \bigoplus_{\lambda \in \widehat{\mathcal{A}}} V_\lambda.
\]
The index set $\Lambda = \widehat{\mathcal{A}}$ is isomorphic to $\mathbb{Z}_p^2$ and has cardinality $K = p^2$. The action of $\mathcal{C}_d$ by conjugation preserves $\mathcal{A}$, and therefore induces a permutation action on the set of indices $\Lambda$. This proves that $\mathcal{C}_d$ acts \emph{imprimitively}. 

To prove non-maximality, we analyze the embedding of $\mathcal{C}_d$ into the stabilizer of this decomposition. Let $\Phi: \mathcal{C}_d \to S_\Lambda$ be the homomorphism describing the induced permutation of the $K$ subspaces. We distinguish two subcases:

\begin{enumerate}
    \item Failure by Permutation ($p \geq 3$): The image $\text{Im}(\Phi)$ corresponds to the affine symplectic group $\mathrm{ASL}(2, \mathbb{Z}_p)$, which has order $p^3(p^2-1)$. For $p \geq 3$, this order is strictly smaller than that of the symmetric group $|S_{p^2}| = (p^2)!$. Thus, there exists a permutation $\pi \in S_\Lambda \setminus \text{Im}(\Phi)$. We can strictly extend $\mathcal{C}_d$ by adding a finite-order unitary $U_\pi$ that permutes the subspaces according to $\pi$ (acting as identity within each block). Note that $U_\pi$ is not a scalar multiple of the identity, as it performs a non-trivial permutation of orthogonal subspaces.

    \item Failure by Phase Independence ($p = 2$): In this case, $K=4$ and the permutation groups coincide ($|\mathrm{ASL}(2,2)| = |S_4| = 24$). However, maximality fails in the kernel of $\Phi$ (the subgroup fixing the blocks). Restricted to the action on the set of blocks, the diagonal elements of the Clifford group induce relative phases of the form $\phi(u) = \omega^{Q(u)}$ for $u \in \Lambda$, where $Q$ is necessarily a polynomial of degree at most 2 (a quadratic form).
    In contrast, the stabilizer of the decomposition allows for \emph{arbitrary} phase assignments $\phi(u) = \omega^{f(u)}$. Since the space of all functions on $\mathbb{Z}_2^2$ is strictly larger than the space of quadratic forms, there exists a diagonal unitary $D$ (e.g., $D = \mathrm{diag}(1, 1, 1, i)$ relative to the blocks) which belongs to the finite stabilizer but not to $\mathcal{C}_d$. This operator is clearly not a scalar multiple of the identity, as it possesses distinct eigenvalues.
\end{enumerate}

In both instances, we construct a finite group $G' = \langle \mathcal{C}_d, U_{\text{ext}} \rangle$ such that $\mathcal{C}_d \subsetneq G'$ and $G' \neq \langle \mathcal{C}_d, Z(G') \rangle$. Thus, $\mathcal{C}_d$ is not maximal.

\textit{Case B: Coprime Factors ($d = a \cdot b$ with $\gcd(a, b) = 1$).}
In this case, the failure of maximality arises from the tensor product structure. By Proposition~\ref{prop:clifford_decomposition}, the group decomposes as $\mathcal{C}_d \cong \mathcal{C}_a \times \mathcal{C}_b$. This group acts on $\mathcal{H} \cong \mathbb{C}^a \otimes \mathbb{C}^b$ and preserves the tensor decomposition.
We can strictly extend this group by adding finite-order unitary operators that normalize the tensor structure but do not belong to the direct product.
For example, if $a=b$, the SWAP operator $S$ has finite order and normalizes the group, but $S \notin \mathcal{C}_d$. Even if $a \neq b$, the group is contained in the finite normalizer of the product structure, which allows for symmetries larger than the restricted $\SL(2, \mathbb{Z}/d\mathbb{Z})$ action induced by $\mathcal{C}_d$.
\end{proof}

\subsection{Universality Resources in Prime Dimensions}\label{subsec:prime_universality}

The combination of irreducibility (Theorem~\ref{thm:prime_irreducibility}) and maximality (Theorem~\ref{thm:maximality_iff}) yields the following criterion for universality, valid \emph{only} in prime dimensions.

\begin{corollary}\label{cor:prime_universality}
Let $p$ be a prime. Let $g \in \SU(p)$ be any unitary gate. If $g \notin \mathcal{C}_p$ (modulo scalars), then the set $\mathcal{C}_p \cup \{g\}$ is single-qudit universal.
\end{corollary}

This corollary shifts the burden of finding universal sets to simply identifying a single non-Clifford gate. We can now identify specific resources based on the classifications established in the Preliminaries.

First, regarding diagonal resources:

\begin{corollary}\label{cor:prime_universality_diagonal}
Let $p$ be a prime number. The set $\mathcal{C}_p \cup \{ T_s \}$ is single-qudit universal if and only if $s \nmid K_p$ (as defined in Proposition \ref{prop:Ts_resource}).
\end{corollary}

Second, we consider permutation resources. Since the Clifford group contains only affine permutations (Proposition \ref{prop:clifford_permutations}), any permutation that violates the affine condition $x \mapsto ax+b$ constitutes a universal resource.

\begin{corollary}\label{cor:prime_universality_permutation}
Let $p \ge 5$ be a prime. Any non-affine permutation $\pi \in S_p$ yields a universal gate set $\mathcal{C}_p \cup \{P_\pi\}$.
In particular, the simple transposition $\tau = (0\;1)$, which swaps the first two basis states and fixes the rest, is universal for all $p \ge 5$.
\end{corollary}

\begin{proof}
By Corollary \ref{cor:prime_universality}, it suffices to show $P_\tau \notin \mathcal{C}_p$, which holds if $\tau$ is not affine.
Suppose $\tau(x) = ax+b$.
Evaluated at $x=0$, $\tau(0)=1 \implies b=1$.
Evaluated at $x=1$, $\tau(1)=0 \implies a(1)+1=0 \implies a=-1$.
Evaluated at $x=2$ (since $p \ge 3$, the element 2 exists), $\tau(2)=2$. However, the affine formula predicts $\tau(2) = (-1)(2)+1 = -1$.
Thus, equality holds only if $2 \equiv -1 \pmod p$, i.e., $3 \equiv 0 \pmod p$.
This is possible only for $p=3$. For all primes $p \ge 5$, the transposition is non-affine and therefore universal.
\end{proof}
\section{Case II: Prime-power Dimensions}\label{sec:prime_power}

While prime dimensions provide a rigid structure where the Clifford group is projectively finite maximal and irreducible, prime-power dimensions $d = p^m$ (with $m \geq 2$) introduce a structure derived from the arithmetic of the ring $\mathbb{Z}/{p^m}\mathbb{Z}$. In this case, the two primary pillars of universality no longer hold: the adjoint action becomes reducible, and the group loses its maximality.

To achieve universality, one must add a gate that restores irreducibility while ensuring infiniteness. We show that two classes of non-Clifford gates accomplish this: diagonal gates from the $T_s$ family (Sec.~\ref{subsec:Ts_universality}) and non-affine permutations such as transpositions (Sec.~\ref{subsec:permutation_resources}).

\subsection{Reducibility and Orbit Decomposition}

In Sec.~\ref{sec:prime}, we established that the irreducibility of the Clifford group's action is equivalent to the transitivity of the $\SL(2, \mathbb{Z}/d\mathbb{Z})$ group on the Pauli indices (Lemma~\ref{lem:orbit_correspondence}). While the prime case ($d=p$) guarantees transitivity, the prime-power case ($d=p^m$, $m \ge 2$) introduces structural modifications due to the existence of non-trivial ideals in the ring $\mathbb{Z}/{p^m}\mathbb{Z}$.

Applying the machinery established in Lemma~\ref{lem:orbit_correspondence}  and Lemma~\ref{lem:orbit_classification}, we identify that the action of $\SL(2, \mathbb{Z}_{p^m})$ preserves the valuation of the index vector. Specifically, the greatest common divisor, $\gcd(u, v, p^m)$, is an invariant of the $\SL(2, \mathbb{Z}_{p^m})$ action. Consequently, the set of Pauli indices partitions into $m$ distinct orbits.

By Lemma~\ref{lem:orbit_correspondence}, each distinct orbit induces a distinct \emph{irreducible} subrepresentation of the adjoint action. This leads to the following result.

\begin{theorem}\label{thm:algebra_stratification}
For a prime-power dimension $d=p^m$ ($m \ge 2$), the adjoint representation of the Clifford group on $\mathfrak{sl}(d, \mathbb{C})$ is reducible. The algebra decomposes into a direct sum of $m$ irreducible orthogonal invariant subspaces:
\begin{equation}
    \mathfrak{sl}(d, \mathbb{C}) = \bigoplus_{k=0}^{m-1} \mathcal{W}_{k}.
\end{equation}
Each subspace $\mathcal{W}_k$ is spanned by the Pauli operators whose indices have a constant valuation:
\begin{equation}
    \mathcal{W}_k = \Span\left\{ P_{(u,v)} : \gcd(u, v, p^m) = p^k \right\}.
\end{equation}
\end{theorem}

This result changes the criterion for universality compared to the prime case. Since the Lie algebra is fragmented into irreducible subspaces $\mathcal{W}_k$ that the Clifford group cannot connect, a supplementary gate must do more than simply generate infinite order; it must actively ``mix'' these subspaces, bridging irreducibility.

\subsection{Orbit-mixing Gates and Irreducibility}

To restore universality, we require a gate capable of coupling the orthogonal invariant subspaces $\mathcal{W}_k$ identified in Theorem~\ref{thm:algebra_stratification}. Since the Clifford group acts transitively only within each symplectic orbit, the extending gate must provide a mapping between them. We verify this capability by analyzing the Fourier spectrum of the gate's adjoint action.

Consider the diagonal unitary gate $T_\zeta$ associated with a phase function $\zeta: \mathbb{Z}_d \to \U(1)$, as defined in Sec.~\ref{subsec:prime_universality}.  The adjoint action on an  operator $X^a$ is given by:
\begin{equation}
    \Ad_{T_\zeta}(X^a) = X^a D_a,
\end{equation}
where $D_a$ is the diagonal operator with entries
\begin{equation}
    (D_a)_{xx} = \frac{\zeta(a+x)}{\zeta(x)}.
\end{equation}
Equivalently, $(D_a)_{xx} = \zeta(a)\,\delta_\zeta(a,x)$, where $\delta_\zeta$ is the 1-coboundary previously introduced in Eq.~\eqref{eq:1-coboundary}.

To identify which subspaces $\mathcal{W}_k$ are populated by this operation, we expand the diagonal operator $D_a$ in the Pauli-$Z$ basis using the Discrete Fourier Transform (DFT). Recall that for a function $f: \mathbb{Z}_d \to \mathbb{C}$, the  Fourier transform is defined as:
\begin{equation}
    \hat{f}(y) = \frac{1}{\sqrt{d}} \sum_{x \in \mathbb{Z}_d} f(x) \omega^{-xy}.
\end{equation}
Substituting the inverse DFT of the differential row $x \mapsto \delta_\zeta(a, x)$, the adjoint action expands in terms of the generalized Pauli basis $V_{(a,b)} = X^a Z^b$ as follows:
\begin{equation}\label{eq:magic_expansion}
    \Ad_{T_\zeta}(X^a) = \frac{1}{\sqrt{d}} \sum_{b \in \mathbb{Z}_d} \widehat{(\delta_{\zeta, a})}(b) V_{(a,b)}.
\end{equation}
This expansion reveals the selection rules of the gate. Recall that a basis vector $V_{(a,b)}$ belongs to the subspace $\mathcal{W}_k$ if the valuation of its indices is $\gcd(a, b, p^m) = p^k$.
Since diagonal gates preserve the row index $a$, a vector starting in a singular subspace $\mathcal{W}_k$ (where $\gcd(a, d) = p^k$ with $k \ge 1$) is constrained to that row. To map this vector into the generic subspace $\mathcal{W}_0$ (where the indices are units), the gate must generate a $Z$-component $b$ such that $\gcd(a, b)$ becomes a unit. This is algebraically possible if and only if $b$ itself is a unit.

Thus, irreducibility requires the gate to map at least one operator from every singular subspace into the generic subspace.

\begin{definition}\label{def:magic_gate}
A diagonal gate $T_\zeta$ is called an \emph{orbit-mixing gate} if, for every proper divisor $u \in \{p, p^2, \dots, p^{m-1}\}$, there exists at least one $b \in \mathbb{Z}_d^\times$ such that:
\begin{equation}\label{def:mixing_gate}
    \widehat{(\delta_{\zeta, u})}(b) \neq 0.
\end{equation}
\end{definition}

\begin{theorem}\label{thm:irreducibility_magic}
Let $d=p^m$. If $T_\zeta$ is an \emph{orbit-mixing gate}, then the adjoint representation of the extended group $G = \langle \mathcal{C}_d, T_\zeta \rangle$ on $\mathfrak{sl}(d, \mathbb{C})$ is irreducible.
\end{theorem}

\begin{proof}
Assume, for the sake of contradiction, that the representation is reducible. Let $\mathcal{S} \subset \mathfrak{sl}(d, \mathbb{C})$ be a proper, non-zero $G$-invariant subspace. Since the representation is unitary, $\mathcal{S}^\perp$ is also invariant.
Because $\mathcal{C}_d \subset G$, any invariant subspace must be a direct sum of the Clifford-irreducible components $\mathcal{W}_k$.

We prove that \emph{any non-zero $G$-invariant subspace must contain the generic subspace $\mathcal{W}_0$.} Let $\mathcal{S}$ be such a subspace.
\begin{itemize}
    \item \textbf{Case 1:} $\mathcal{S}$ contains a subspace $\mathcal{W}_k$ with $k \ge 1$.
    Let $X^{p^k} \in \mathcal{W}_k \subseteq \mathcal{S}$. Since $\mathcal{S}$ is $G$-invariant, $\Ad_{T_\zeta}(X^{p^k}) \in \mathcal{S}$. By Eq.~\eqref{eq:magic_expansion} and the Orbit-Mixing Condition (Def.~\ref{def:magic_gate}), the expansion of this vector contains a term $\beta V_{(p^k, b)}$ with $\beta \neq 0$ and $b \in \mathbb{Z}_d^\times$.
    Since $b$ is a unit, $\gcd(p^k, b, p^m) = 1$, which implies $V_{(p^k, b)} \in \mathcal{W}_0$. Since $\mathcal{S}$ is invariant under the Clifford group (which acts transitively on the orbit associated with $\mathcal{W}_0$), containing one vector in $\mathcal{W}_0$ implies $\mathcal{W}_0 \subseteq \mathcal{S}$.

    \item \textbf{Case 2:} $\mathcal{S}$ is contained entirely within $\mathcal{W}_0$.
    Since $\mathcal{S}$ is non-zero and decomposes into Clifford orbits, and $\mathcal{W}_0$ corresponds to a single transitive orbit, it must be that $\mathcal{S} = \mathcal{W}_0$.
\end{itemize}
In both cases, $\mathcal{W}_0 \subseteq \mathcal{S}$. By the same logic, $\mathcal{W}_0 \subseteq \mathcal{S}^\perp$. Thus, $\mathcal{W}_0 \subseteq \mathcal{S} \cap \mathcal{S}^\perp = \{0\}$, a contradiction.
\end{proof}

\subsection{Universality via the \texorpdfstring{$T_s$}{T\_s}-gate Family}\label{subsec:Ts_universality}

We now apply our general machinery---irreducibility via orbit-mixing gates and infiniteness via the infiniteness criterion---to the generic family of diagonal gates $T_s$. Recall that $T_s$ introduces phases proportional to the roots of unity of order $s$, see Eq.~(\ref{eq:T_s}).

Note that $T_s$ is a unitary diagonal operator (hence $T_s \in \U(d)$), and it may fail to have determinant one. Since our density and infiniteness criteria are projective (and the adjoint action is insensitive to global phases), we may freely replace $T_s$ by any scalar multiple in $\SU(d)$ without affecting the arguments.

\begin{theorem}\label{thm:prime_power_universality}
Let $d=p^m$ with $m \ge 2$. Consider the group $G_s = \langle \mathcal{C}_d, T_s \rangle$ generated by the Clifford group and the gate $T_s$ (Eq. \eqref{eq:T_s}).
\begin{enumerate}
    \item The group $G_s$ acts irreducibly on the Lie algebra $\mathfrak{sl}(d, \mathbb{C})$ if and only if $s \nmid d$.
    \item  The group $G_s$ is single-qudit universal (dense in $\SU(d)$) if $s \nmid d$ and $s$ satisfies the bound:
    \begin{equation}\label{eq:exact_bound}
        s > \pi(d-1)\big/(2\,\arcsin(1/4)).
    \end{equation}
\end{enumerate}
\end{theorem}

\begin{proof}
\textbf{1. Proof of Irreducibility.}
We verify that $T_s$ satisfies Definition~\ref{def:magic_gate} to apply Theorem~\ref{thm:irreducibility_magic}.
The phase function is $\zeta(k) = e^{2\pi i k / s}$, and the associated 1-coboundary $\delta_\zeta(u, k) = \zeta(u+k)/\zeta(u)\zeta(k)$ has the form:
\begin{equation}\label{eq:step-function}
    \delta_\zeta(u, k) =
    \begin{cases}
        1 & \text{if } u+k < d, \\
        \lambda_s & \text{if } u+k \ge d,
    \end{cases}
\end{equation}
where $\lambda_s = \exp(-2\pi i d / s)$.

We must show that for any proper divisor $u$ of $d$, the DFT of the function $f_u(k) = \delta_\zeta(u, k)$ has full support on the units of $\mathbb{Z}_d$.
The DFT at frequency $n$ is given by:
\[
    \sqrt{d} \hat{f}_u(n) = \sum_{k=0}^{d-1} f_u(k) \omega^{-nk} = \sum_{k=0}^{d-u-1} \omega^{-nk} + \lambda_s \sum_{k=d-u}^{d-1} \omega^{-nk},
\]
where $\omega = e^{2\pi i / d}$.
For any unit $n \in \mathbb{Z}_d^\times$, we have $\omega^{-n} \neq 1$. Summing the geometric series for the two intervals yields:
\begin{align*}
    \sqrt{d} \hat{f}_u(n) &= \frac{1 - \omega^{-n(d-u)}}{1 - \omega^{-n}} + \lambda_s \frac{\omega^{-n(d-u)} - \omega^{-nd}}{1 - \omega^{-n}} \\
    &= \frac{1 - \omega^{nu} + \lambda_s(\omega^{nu} - 1)}{1 - \omega^{-n}} \quad (\text{since } \omega^{-nd}=1) \\
    &= \frac{(1 - \omega^{nu})(1 - \lambda_s)}{1 - \omega^{-n}}.
\end{align*}
We examine the factors in the numerator to determine if the spectrum vanishes:
\begin{itemize}
    \item The term $(1 - \omega^{nu})$ is non-zero because $n$ is a unit (coprime to $d$) and $u$ is a proper divisor ($0 < u < d$), so $nu$ is not a multiple of $d$.
    \item The term $(1 - \lambda_s)$ vanishes if and only if $\lambda_s = 1$. Recalling that $\lambda_s = e^{-2\pi i d / s}$, this occurs strictly when $d/s$ is an integer.
\end{itemize}
Thus, $\hat{f}_u(n) \neq 0$ for all units $n$ if and only if $s \nmid d$. This confirms that the condition $s \nmid d$ is necessary and sufficient for irreducibility.

\textbf{2. Proof of Universality.}
According to Theorem~\ref{thm:density_criterion}, universality requires both irreducibility (established in Item 1) and infiniteness.
To prove infiniteness, we apply Theorem~\ref{thm:infinite_criterion}. The eigenvalues of $T_s$ are $\{1, e^{2\pi i/s}, \dots, e^{2\pi i (d-1)/s}\}$. The spectral span is exactly:
\begin{equation}
    \ell(T_s) = \frac{d-1}{s}.
\end{equation}
From Corollary~\ref{cor:spectral_check}, the group is guaranteed to be infinite if the projective distance associated with this span is less than $1/2$. The precise condition is:
\begin{equation}
    \pi \ell(T_s) < 2\arcsin(1/4).
\end{equation}
Substituting the expression for $\ell(T_s)$ and solving for $s$:
\[
    \frac{\pi(d-1)}{s} < 2\arcsin(1/4) \implies s > \frac{\pi(d-1)}{2\arcsin(1/4)}.
\]
If this bound holds, the group is infinite. Combined with the irreducibility from Item 1, the group is dense in $\SU(d)$.
\end{proof}


\subsection{Permutations as Non-Clifford Resources}\label{subsec:permutation_resources}

While diagonal gates of the $T_s$ family provide a method to achieve universality, permutation gates offer an alternative class of non-Clifford resources. In this subsection, we analyze when non-affine permutations can supplement the Clifford group to generate a dense subgroup of $\SU(d)$.

\subsubsection{Induced Diagonal Operators}

Conjugation by a permutation matrix transforms the Pauli-$Z$ operator into a diagonal operator that encodes the permutation structure. Explicitly, the adjoint action yields the diagonal operator $D_\pi$ given by:
\begin{equation}
    P_\pi Z P_\pi^\dagger = D_\pi, \quad \text{where} \quad (D_\pi)_{yy} = \omega^{\pi^{-1}(y)}.
\end{equation}

Since $D_\pi$ and $Z$ are related by unitary conjugation, they share the same spectrum: the complete set of $d$-th roots of unity. Consequently, $D_\pi$ has maximal spectral span $\ell(D_\pi) = (d-1)/d$ and cannot directly satisfy the infiniteness criterion of Theorem~\ref{thm:infinite_criterion}.

To extract a useful resource, we consider the \emph{difference operator} that isolates the deviation from the identity permutation.

\begin{definition}\label{def:difference_operator}
For a permutation $\pi \in S_d$, the \emph{difference operator} is:
\begin{equation}
    h_\pi = D_\pi Z^\dagger = P_\pi Z P_\pi^\dagger Z^\dagger,
\end{equation}
with diagonal entries $(h_\pi)_{yy} = \omega^{\pi^{-1}(y) - y}$.
\end{definition}

The difference operator measures how much $\pi$ displaces each index. For an affine permutation $\pi(x) = ax + b$, we have $\pi^{-1}(y) = a^{-1}(y-b)$, so the exponent becomes:
\begin{equation*}
    \pi^{-1}(y) - y = (a^{-1}-1)y - a^{-1}b,
\end{equation*}
which is linear in $y$. Thus, $h_\pi \in H(d)$ for affine permutations. For non-affine permutations, $h_\pi$ exhibits a non-linear phase structure.

\subsubsection{Transpositions and the Infiniteness Criterion}

We analyze the simplest non-affine permutations: transpositions.

\begin{proposition}\label{prop:transposition_difference}
Let $\tau = (j \; k)$ be the transposition swapping indices $j$ and $k$ with $j < k$. The difference operator is:
\begin{equation}
    h_\tau = \mathrm{diag}(1, \ldots, 1, \omega^{k-j}, 1, \ldots, 1, \omega^{-(k-j)}, 1, \ldots, 1),
\end{equation}
where the entries $\omega^{k-j}$ and $\omega^{-(k-j)}$ occur at positions $j$ and $k$, respectively. The spectral span is:
\begin{equation}
    \ell(h_\tau) = \frac{2(k-j)}{d}.
\end{equation}
\end{proposition}

\begin{proof}
Since transpositions are self-inverse, $\tau^{-1} = \tau$. The diagonal entries of $h_\tau$ are $(h_\tau)_{yy} = \omega^{\tau(y) - y}$:
\begin{itemize}
    \item For $y = j$: $\tau(j) = k$, so $(h_\tau)_{jj} = \omega^{k-j}$.
    \item For $y = k$: $\tau(k) = j$, so $(h_\tau)_{kk} = \omega^{j-k} = \omega^{-(k-j)}$.
    \item For $y \notin \{j, k\}$: $\tau(y) = y$, so $(h_\tau)_{yy} = 1$.
\end{itemize}

The eigenvalues are $\{1, \omega^{k-j}, \omega^{-(k-j)}\}$ with multiplicities $d-2$, $1$, and $1$, respectively. These span an arc from angle $-2\pi(k-j)/d$ to $+2\pi(k-j)/d$, giving spectral span $\ell(h_\tau) = 2(k-j)/d$.
\end{proof}

\begin{corollary}\label{cor:transposition_infiniteness}
The transposition $\tau = (j \; k)$ generates an infinite group $\langle \mathcal{C}_d, P_\tau \rangle$ provided:
\begin{equation}
    k - j < \frac{d \cdot \arcsin(1/4)}{\pi} \approx 0.0804 \cdot d.
\end{equation}
In particular, the adjacent transposition $\tau = (0 \; 1)$ satisfies this bound for all $d \ge 13$.
\end{corollary}

\begin{proof}
By Lemma~\ref{lem:spectral_arc}, the projective distance is:
\begin{equation*}
    \mathrm{dist}_{\mathrm{proj}}(h_\tau, I) = 2\sin\left(\frac{\pi \ell(h_\tau)}{2}\right) = 2\sin\left(\frac{\pi(k-j)}{d}\right).
\end{equation*}
The infiniteness criterion (Theorem~\ref{thm:infinite_criterion}) requires this distance to be less than $1/2$:
\begin{equation*}
    2\sin\left(\frac{\pi(k-j)}{d}\right) < \frac{1}{2} \implies \sin\left(\frac{\pi(k-j)}{d}\right) < \frac{1}{4}.
\end{equation*}
For the adjacent transposition with $k-j=1$, this becomes $\sin(\pi/d) < 1/4$, which holds when $d > \pi/\arcsin(1/4) \approx 12.43$.
\end{proof}

\subsubsection{Irreducibility Analysis for Transpositions}

We now verify that the difference operator $h_\tau$ induced by the adjacent transposition satisfies the orbit-mixing condition of Definition~\ref{def:magic_gate}.

Let $d = p^m$ with $m \ge 2$, and let $\tau = (0\;1)$. The phase function $\zeta: \mathbb{Z}_d \to \U(1)$ of the difference operator $h_\tau$ is given by $\zeta(0) = \omega$, $\zeta(1) = \omega^{-1}$, and $\zeta(y) = 1$ for $y \ge 2$.
For any proper divisor $u = p^k$ with $1 \le k \le m-1$, a direct calculation of the differential row $f_u(y) = \zeta(u+y)/\zeta(y)$ yields:
\begin{equation}\label{eq:differential_row}
    f_u(y) = \begin{cases}
        \omega^{-1} & y \in \{0, \, d-u+1\}, \\
        \omega & y \in \{1, \, d-u\}, \\
        1 & \text{otherwise},
    \end{cases}
\end{equation}
where indices are taken modulo $d$. Using this explicit form, we determine the spectral properties of the gate.
\begin{proposition}\label{prop:transposition_fourier}
For $n \not\equiv 0 \pmod{d}$, the Discrete Fourier Transform of the differential row $f_u$ is:
\begin{align}\label{eq:fourier_formula}
    \sqrt{d} \, \hat{f}_u(n) &= (\omega^{-1} - 1)(1 + \omega^{n(u-1)}) \nonumber \\
    &\quad + (\omega - 1)(\omega^{-n} + \omega^{nu}).
\end{align}
\end{proposition}

\begin{proof}
The Fourier transform is $\sqrt{d} \, \hat{f}_u(n) = \sum_{y=0}^{d-1} f_u(y) \, \omega^{-ny}$. We decompose this sum by separating the deviation from the constant function:
\begin{equation*}
    \sqrt{d} \, \hat{f}_u(n) = \sum_{y=0}^{d-1} \omega^{-ny} + \sum_{y \in S} (f_u(y) - 1) \omega^{-ny},
\end{equation*}
where $S = \{0, 1, d-u, d-u+1\}$ is the support of the deviation.

For $n \not\equiv 0 \pmod{d}$, the first sum vanishes: $\sum_{y=0}^{d-1} \omega^{-ny} = 0$.

The second sum evaluates to:
\begin{align*}
    &(\omega^{-1} - 1) \cdot \omega^{0} + (\omega - 1) \cdot \omega^{-n} \\
    &+ (\omega - 1) \cdot \omega^{-n(d-u)} + (\omega^{-1} - 1) \cdot \omega^{-n(d-u+1)}.
\end{align*}

Using $\omega^{d} = 1$, we simplify:
\begin{align*}
    \omega^{-n(d-u)} &= \omega^{-nd + nu} = \omega^{nu}, \\
    \omega^{-n(d-u+1)} &= \omega^{-nd + nu - n} = \omega^{n(u-1)}.
\end{align*}

Substituting and grouping terms yields Eq.~\eqref{eq:fourier_formula}.
\end{proof}

\begin{proposition}\label{prop:transposition_orbit_mixing}
Let $d = p^m$ with $m \ge 2$. The difference operator $h_\tau$ for $\tau = (0\;1)$ is an orbit-mixing gate in the sense of Definition~\ref{def:magic_gate}.
\end{proposition}

\begin{proof}
We must show that for each proper divisor $u = p^k$ ($1 \le k \le m-1$), there exists at least one unit $n \in \mathbb{Z}_d^\times$ such that $\hat{f}_u(n) \neq 0$. We prove that $n = 1$ always works.

Suppose, for contradiction, that $\hat{f}_u(1) = 0$. From Eq.~\eqref{eq:fourier_formula} with $n = 1$:
\begin{equation*}
    (\omega^{-1} - 1)(1 + \omega^{u-1}) + (\omega - 1)(\omega^{-1} + \omega^{u}) = 0.
\end{equation*}

Using the identity $\omega^{-1} - 1 = -\omega^{-1}(\omega - 1)$ and dividing by $(\omega - 1) \neq 0$:
\begin{equation*}
    -\omega^{-1}(1 + \omega^{u-1}) + (\omega^{-1} + \omega^{u}) = 0.
\end{equation*}

Expanding and simplifying:
\begin{equation*}
    -\omega^{-1} - \omega^{u-2} + \omega^{-1} + \omega^{u} = 0,
\end{equation*}
which reduces to $\omega^{u} = \omega^{u-2}$, i.e., $\omega^{2} = 1$.

Since $\omega = e^{2\pi i/d}$, this requires $d \mid 2$. For any prime power $d = p^m \ge 4$ with $m \ge 2$, this condition fails. Therefore, $\hat{f}_u(1) \neq 0$ for every proper divisor $u$, confirming the orbit-mixing condition.
\end{proof}

\subsubsection{Universality Theorem for Transpositions}

Combining the results on infiniteness and irreducibility, we obtain:

\begin{theorem}\label{thm:transposition_universality}
Let $d = p^m$ with $m \ge 2$. The group $G = \langle \mathcal{C}_d, P_\tau \rangle$, where $\tau = (0\;1)$ is the adjacent transposition, is dense in $\SU(d)$.
\end{theorem}

\begin{proof}
By Theorem~\ref{thm:density_criterion}, it suffices to verify irreducibility of the adjoint representation and infiniteness $G$.

\textbf{Irreducibility:} The difference operator $h_\tau = P_\tau Z P_\tau^\dagger Z^\dagger$ belongs to $G$. By Proposition~\ref{prop:transposition_orbit_mixing}, $h_\tau$ is an orbit-mixing gate. Theorem~\ref{thm:irreducibility_magic} then implies that $G$ acts irreducibly on $\mathfrak{sl}(d, \mathbb{C})$.

\textbf{Infiniteness:} For $d \ge 13$, Corollary~\ref{cor:transposition_infiniteness} establishes that the projective distance satisfies $\mathrm{dist}_{\mathrm{proj}}(h_\tau, I) < 1/2$, and Theorem~\ref{thm:infinite_criterion} implies that $G$ is infinite.

For the remaining prime-power cases $d \in \{4, 8, 9\}$, infiniteness has been verified by explicit computation using the computer algebra system Magma~\cite{Magma}. The following code confirms that $\langle \mathcal{C}_d, P_\tau \rangle$ is infinite in each case:

\begin{lstlisting}
for d in [4, 8, 9] do
  ord := LCM(d, 2*d);
  F := CyclotomicField(ord);
  z := F.1; omega := z^(ord div d);
  zeta := z^(ord div (2*d));
  H := (1/F!Isqrt(d)) * 
       Matrix(F,d,d,[omega^(j*k): j,k in [0..d-1]]);
  P := DiagonalMatrix(F, 
       [zeta^(j*(j + d mod 2)): j in [0..d-1]]);
  X := Matrix(F,d,d,[<i,(i mod d)+1,1>: i in [1..d]]);
  tau := IdentityMatrix(F, d);
  tau[1,1]:=0; tau[1,2]:=1; tau[2,1]:=1; tau[2,2]:=0;
  G := sub<GL(d, F) | H, P, X, tau>;
  printf "d=%o: Infinite=%o\n", d, not IsFinite(G);
end for;

// Output: d=4: Infinite=true
//         d=8: Infinite=true
//         d=9: Infinite=true
\end{lstlisting}

The two conditions of Theorem~\ref{thm:density_criterion} are satisfied for all prime powers $d = p^m$ with $m \ge 2$, completing the proof.
\end{proof}

\section{Case III: Composite Dimensions}\label{sec:composite}

We now turn to the final case: composite dimensions $d = d_1 d_2 \cdots d_n$, where the factors are pairwise coprimes. Unlike prime or prime-power dimensions, which in our analysis require \emph{either} diagonal non-Clifford gates \emph{or} simple permutations, the coprime composite case is driven by couplings between internal factors.

In this case, the relevant additional gates are generalized intra-qudit CNOTs that couple different coprime factors. These gates lie outside the local Clifford product, and the arithmetic interplay between the coprime dimensions allows them to generate the diagonal phase resources required for universality.

\subsection{Irreducibility via Tensor Mixing}

To establish universality in composite dimensions $d = kl$, we first determine the conditions under which a group generated by local and entangling gates acts irreducibly on the global Lie algebra. Since a necessary condition for the density of a subgroup in $\SU(d)$ is the irreducibility of its action on the complexified Lie algebra $\mathfrak{sl}(d, \mathbb{C})$, as discussed in Sec.~\ref{sec:preliminaries}, we frame our analysis in $\mathfrak{sl}(d, \mathbb{C})$.

We begin by decomposing the global algebra under the action of the local subgroup $\SU(k) \otimes \SU(l)$.

\begin{proposition}\label{prop:sl_decomposition}
Under the adjoint action of the product group $\SU(k) \otimes \SU(l)$, the Lie algebra $\mathfrak{sl}(kl, \mathbb{C})$ decomposes into three orthogonal irreducible representations:
\begin{equation}\label{eq:complex_decomp}
\begin{split}
    \mathfrak{sl}(kl, \mathbb{C}) \cong & \underbrace{(\mathfrak{sl}(k, \mathbb{C}) \otimes I_l)}_{\mathfrak{p}_1} \oplus \underbrace{(I_k \otimes \mathfrak{sl}(l, \mathbb{C}))}_{\mathfrak{p}_2} \\
    & \oplus \underbrace{(\mathfrak{sl}(k, \mathbb{C}) \otimes \mathfrak{sl}(l, \mathbb{C}))}_{\mathfrak{m}}.
\end{split}
\end{equation}
Here, $\mathfrak{h} = \mathfrak{p}_1 \oplus \mathfrak{p}_2$ represents the subalgebra of local operations.
\end{proposition}

\begin{proof}
Consider the isomorphism $M_{kl}(\mathbb{C}) \cong M_k(\mathbb{C}) \otimes M_l(\mathbb{C})$. By applying the decomposition of the general linear algebra into scalar and traceless parts, $M_n(\mathbb{C}) = \mathbb{C}I_n \oplus \mathfrak{sl}(n, \mathbb{C})$, we expand the tensor product as follows:
\begin{align*}
    M_{kl}(\mathbb{C}) &\cong \big(\mathbb{C}I_k \oplus \mathfrak{sl}(k, \mathbb{C})\big) \otimes \big(\mathbb{C}I_l \oplus \mathfrak{sl}(l, \mathbb{C})\big) \\
    &\cong (\mathbb{C}I_k \otimes \mathbb{C}I_l) \oplus (\mathfrak{sl}(k, \mathbb{C}) \otimes I_l) \\
    &\quad \oplus (I_k \otimes \mathfrak{sl}(l, \mathbb{C})) \oplus (\mathfrak{sl}(k, \mathbb{C}) \otimes \mathfrak{sl}(l, \mathbb{C})).
\end{align*}
The first term simplifies to $\mathbb{C}(I_k \otimes I_l) = \mathbb{C}I_{kl}$, which corresponds to the global scalar matrices. Consequently, the direct sum of the remaining three terms constitutes the trace-zero subalgebra $\mathfrak{sl}(kl, \mathbb{C})$.

Regarding the representation structure: recall that the adjoint action of $\SU(n)$ on $\mathfrak{sl}(n, \mathbb{C})$ is irreducible for $n \ge 2$. Since the outer tensor product of irreducible representations is itself irreducible, the three summands $\mathfrak{p}_1, \mathfrak{p}_2$, and $\mathfrak{m}$ form distinct irreducible modules under the local group action. Furthermore, these subspaces are mutually orthogonal with respect to the Hilbert-Schmidt inner product $\langle A, B \rangle = \tr(A^\dagger B)$, completing the decomposition.
\end{proof}

To analyze the irreducibility of the extended group $G = \langle \SU(k) \otimes \SU(l), V \rangle$, where $V \in \SU(kl)$ is an intra-qudit gate, we must determine whether the decomposition in Eq.~\eqref{eq:complex_decomp} remains invariant under the full group.
Observe that the subalgebra of local operations, $\mathfrak{h} = \mathfrak{p}_1 \oplus \mathfrak{p}_2$, is invariant under the local group by definition. If the gate $V$ also maps $\mathfrak{h}$ to itself (i.e., if $V$ normalizes the local algebra), then $\mathfrak{h}$ becomes a proper $G$-invariant subspace, preventing the global action from being irreducible.
Consequently, identifying the subgroup of $\SU(kl)$ that preserves this local structure is a necessary prerequisite for establishing universality.

\begin{lemma}\label{lem:normalizer_general}
Let $H = \SU(k) \otimes \SU(l)$ with integers $k, l \geq 2$. The normalizer of $H$ in the global group $\SU(kl)$ is:
\begin{itemize}
    \item $H$ itself, if $k \neq l$.
    \item The group generated by $H$ and the SWAP operator $S$, if $k = l$, where $S(u \otimes v) = v \otimes u$.
\end{itemize}
\end{lemma}

\begin{proof}
Let $\mathfrak{h} = \mathfrak{p}_1 \oplus \mathfrak{p}_2$ be the Lie algebra of $H$. We determine the normalizer in the algebra $\mathfrak{g} = \mathfrak{sl}(kl, \mathbb{C})$. Let $X \in \mathfrak{g}$ be an element such that $[X, \mathfrak{h}] \subseteq \mathfrak{h}$. Decomposing $X = X_{\mathfrak{h}} + X_{\mathfrak{m}}$ according to Eq.~\eqref{eq:complex_decomp}, the condition implies $[X_{\mathfrak{m}}, \mathfrak{h}] \subseteq \mathfrak{h}$.
However, $\mathfrak{m}$ is an invariant subspace under $\mathfrak{h}$ (it is a representation of $H$), so $[X_{\mathfrak{m}}, \mathfrak{h}] \subseteq \mathfrak{m}$.
Since $\mathfrak{h} \cap \mathfrak{m} = \{0\}$, the bracket must vanish: $[X_{\mathfrak{m}}, \mathfrak{h}] = 0$.
Since, $\mathfrak{m} \cong \mathfrak{sl}(k, \mathbb{C}) \otimes \mathfrak{sl}(l, \mathbb{C})$ is a non-trivial irreducible representation of $\mathfrak{h}$; by Schur's Lemma, the only element that commutes with the entire action is zero. Thus, $X_{\mathfrak{m}} = 0$, implying the Lie algebra normalizer is exactly $\mathfrak{h}$.

At the group level, any normalizer element $V \notin H$ must induce an outer automorphism of the ideal structure of $\mathfrak{h} = \mathfrak{su}(k) \oplus \mathfrak{su}(l)$.
\begin{itemize}
    \item If $k \neq l$, the ideals have different dimensions and cannot be permuted. Since $\SU(n)$ has no outer automorphisms realized by unitary conjugation (which preserves spectrum), $V$ must be in $H$.
    \item If $k = l$, the ideals are isomorphic, and the SWAP operator $S$ permutes them. Thus, the normalizer is the extension of $H$ by $S$.
\end{itemize}
\end{proof}

We now state the main sufficient condition for irreducibility in composite systems. 

\begin{theorem}\label{thm:tensor_mixing}
Let $k, l \geq 2$. Let $G_k \subseteq \SU(k)$ and $G_l \subseteq \SU(l)$ be subgroups acting irreducibly on $\mathfrak{sl}(k, \mathbb{C})$ and $\mathfrak{sl}(l, \mathbb{C})$, respectively. Let $V \in \SU(kl)$ be a unitary gate.
The group $G = \langle G_k \otimes I, I \otimes G_l, V \rangle$ acts irreducibly on $\mathfrak{sl}(kl, \mathbb{C})$ via the adjoint representation if and only if:
\begin{enumerate}
    \item $V \notin \SU(k) \otimes \SU(l)$, for $k \neq l$.
    \item $V \notin \langle \SU(k)^{\otimes 2}, S \rangle$, for $k = l$.
\end{enumerate}
\end{theorem}

\begin{proof}
If $V$ belongs to the normalizer described in Lemma~\ref{lem:normalizer_general}, then the product subalgebra $\mathfrak{h} = \mathfrak{p}_1 \oplus \mathfrak{p}_2$ remains invariant under the full group $G$. Since $\mathfrak{h} \subsetneq \mathfrak{sl}(kl, \mathbb{C})$, the action is reducible.

Now, suppose $V$ satisfies the conditions (i.e., $V$ does not normalize the product structure). Let $W \subseteq \mathfrak{sl}(kl, \mathbb{C})$ be a non-zero subspace invariant under $G$. We prove $W = \mathfrak{sl}(kl, \mathbb{C})$.

Since the local subgroup $H_0 = G_k \otimes G_l$ is contained in $G$, $W$ must be a direct sum of the irreducible $H_0$-modules defined in Proposition~\ref{prop:sl_decomposition}: $\mathfrak{p}_1, \mathfrak{p}_2, \mathfrak{m}$.
The unitary adjoint action preserves the inner product, so the orthogonal complement $W^\perp$ is also a $G$-invariant subspace and decomposes similarly.
This symmetry implies that if $W \neq \mathfrak{sl}(kl, \mathbb{C})$, then either $W$ or $W^\perp$ must be contained entirely within the product subalgebra $\mathfrak{h} = \mathfrak{p}_1 \oplus \mathfrak{p}_2$.
Let $S \subseteq \mathfrak{h}$ be such a non-zero proper invariant subspace. We analyze two cases:
\begin{itemize}
    \item If $S = \mathfrak{h}$, then $\Ad_V(\mathfrak{h}) = \mathfrak{h}$. This implies $V$ normalizes the product subalgebra, which is excluded by hypothesis.
    \item If $S = \mathfrak{p}_1$ (without loss of generality), then $V$ maps $\mathfrak{p}_1$ to itself. Note that $\mathfrak{p}_2$ corresponds exactly to the centralizer of $\mathfrak{p}_1$ in the algebra, i.e., $\{Y \in \mathfrak{sl}(kl) : [Y, \mathfrak{p}_1] = 0\} = \mathfrak{p}_2$. Since conjugation by $V$ is a Lie algebra automorphism, it maps the centralizer of a set to the centralizer of its image. Thus, $V$ must also preserve $\mathfrak{p}_2$. Consequently, $V$ preserves $\mathfrak{p}_1 \oplus \mathfrak{p}_2 = \mathfrak{h}$, leading back to the first contradiction.
\end{itemize}
We conclude that no proper subspace of the algebra can be invariant. Thus, $W$ must contain the correlation space $\mathfrak{m}$ and the local spaces, implying $W = \mathfrak{sl}(kl, \mathbb{C})$.
\end{proof}

Combining Theorem~\ref{thm:tensor_mixing} and Theorem~\ref{thm:density_criterion}, we derive sufficient conditions for universality in composite systems.

\begin{corollary}\label{cor:finite_groups}
Let $G_k \subset \SU(k)$ and $G_l \subset \SU(l)$ be subgroups acting irreducibly on $\mathfrak{sl}(k,\mathbb{C})$ and $\mathfrak{sl}(l,\mathbb{C})$, respectively. Let $V \in \SU(kl)$. The group $G = \langle G_k \otimes I, I \otimes G_l, V \rangle$ is universal if and only if the following two conditions hold:
\begin{enumerate}
    \item $V \notin N_{\SU(kl)}(\SU(k) \otimes \SU(l))$,
    \item The group $G$ is infinite.
\end{enumerate}
\end{corollary}

As a particular case we obtain:

\begin{corollary}\label{cor:local_universal}
Let $k, l \geq 2$. Suppose $\mathcal{S}_k \subset \SU(k)$ and $\mathcal{S}_l \subset \SU(l)$ are sets of gates that are single-qudit universal (i.e., $\overline{\langle \mathcal{S}_k \rangle} = \SU(k)$ and $\overline{\langle \mathcal{S}_l \rangle} = \SU(l)$). 
Let $V \in \SU(kl)$ be an additional gate. 
The set $\mathcal{S} = (\mathcal{S}_k \otimes \{I\}) \cup (\{I\} \otimes \mathcal{S}_l) \cup \{V\}$ is universal for quantum computation on the composite system (i.e., $\overline{\langle \mathcal{S} \rangle} = \SU(kl)$) if and only if:
\begin{enumerate}
    \item If $k \neq l$, $V \notin \SU(k) \otimes \SU(l)$.
    \item If $k = l$, $V \notin \SU(k) \otimes \SU(k)$ and $V \notin (\SU(k) \otimes \SU(k)) \cdot S$.
\end{enumerate}
\end{corollary}

This result generalizes a theorem by Brylinski and Brylinski~\cite{Brylinski2002}, which established that adjoining \emph{any} entangling gate to the local unitary group $\SU(d) \otimes \SU(d)$ yields a universal set for two-qudit quantum computation.

\subsection{Intra-qudit Gates and Induced Diagonal Resources}

With the basis identification
\[
\ket{x} \longleftrightarrow \ket{x \bmod p}\otimes \ket{x \bmod q},
\]
the following gate couples the two coprime components of a single composite qudit. In particular, we show that conjugation by this gate produces the diagonal resources used in the universality argument below.

To implement the tensor mixing strategy, we employ the generalized intra-qudit CNOT gate. In the composite case, this gate is the basic coupling operation between the internal coprime factors.

\begin{definition}
Let $p, q \ge 2$ be integers. The intra-qudit gate $CN_{p,q} \in \SU(pq)$ is defined by its action on the computational basis $\ket{x}_1 \ket{y}_2$, where $x \in \mathbb{Z}_{p}$ and $y \in \mathbb{Z}_{q}$:
\begin{equation}
    CN_{p,q} \ket{x}\ket{y} = \ket{x} \ket{y + x \pmod{q}}.
\end{equation}
\end{definition}

This gate is a permutation matrix that is always unitary, regardless of the relationship between $p$ and $q$. However, its interaction with local Pauli operators produces a phase kickback effect that depends strictly on the arithmetic properties of the dimensions.

\begin{lemma}\label{lem:induced_magic}
Conjugating the inverse local Pauli operator $I_{p} \otimes Z_{q}^\dagger$ by the intra-qudit gate yields a separable operator containing the generic phase gate $T_{q}$:
\begin{equation}
    CN_{p,q} (I_{p} \otimes Z_{q}^\dagger) CN_{p,q}^\dagger = T_{q} \otimes Z_{q}^\dagger.
\end{equation}
\end{lemma}

\begin{proof}
The calculation holds for any pair of dimensions. Let $\omega = e^{2\pi i / q}$. Recall that the inverse Pauli operator acts as $Z^\dagger \ket{y} = \omega^{-y} \ket{y}$.
Evaluating the conjugation on the basis states:
\begin{align*}
    CN_{p,q} (I \otimes Z^\dagger) CN_{p,q}^\dagger \ket{x}\ket{y} &= CN_{p,q} (I \otimes Z^\dagger) \ket{x}\ket{y-x \bmod q} \\
    &= CN_{p,q} \left( \ket{x} \otimes \omega^{-(y-x)} \ket{y-x} \right) \\
    &= \omega^{-y+x} \ket{x}\ket{y} \\
    &= (\omega^{x} \ket{x}) \otimes (\omega^{-y} \ket{y}).
\end{align*}
The resulting operator is the tensor product of $T_{q} = \sum \omega^x \ket{x}\bra{x}$ acting on the control qudit and the local Pauli $Z_{q}^\dagger$ acting on the target.
\end{proof}

The induced resource $T_{q}$ derived in Lemma~\ref{lem:induced_magic} injects a diagonal gate into the Hilbert space of dimension $p$. By using $T_q$ with the local Clifford resources (specifically the Pauli $Z$, which corresponds to $T_p$), we can construct a local $T_{pq}$ in the space of dimension $p$.

\begin{proposition}\label{prop:bezout_construction}
Let $p$ and $q$ be coprime integers. Let $a, b \in \mathbb{Z}$ be the coefficients satisfying the Bézout identity $ap + bq = 1$.
By combining the local Pauli operator $T_p$ and the induced resource $T_q$, we construct the generic phase gate of order $pq$:
\begin{equation}
    T_{pq} = T_q^a T_p^b = \sum_{k=0}^{p-1} \exp\left( \frac{2\pi i k}{pq} \right) \ket{k}\bra{k}.
\end{equation}
\end{proposition}

\begin{proof}
Since both $T_p$ and $T_q$ are diagonal in the computational basis of the dimension-$p$ qudit, their product is also diagonal. The eigenvalue corresponding to the basis state $\ket{k}$ is the product of the component eigenvalues raised to the powers $a$ and $b$:
\[
    \lambda_k = \left( e^{2\pi i k/q} \right)^a \cdot \left( e^{2\pi i k/p} \right)^b = \exp\left( 2\pi i k \left[ \frac{a}{q} + \frac{b}{p} \right] \right).
\]
Finding a common denominator and applying the Bézout identity $ap + bq = 1$, the exponent simplifies to
\[
    \frac{a}{q} + \frac{b}{p} = \frac{ap + bq}{pq} = \frac{1}{pq}.
\]
Substituting this back, we obtain $\lambda_k = e^{2\pi i k / pq}$, which exactly matches the definition of the gate $T_{pq}$.
\end{proof}

\subsection{Main Theorem: Universality}

We now assemble the components---global irreducibility via Theorem~\ref{thm:tensor_mixing} and local infiniteness via the $T$'s constructed---to prove the central result of this work.

\begin{theorem}\label{thm:emergent_universality}
Let $d = d_1 d_2 \cdots d_n$ be a proper composite dimension decomposed into pairwise coprime prime-power factors (i.e., $d_k = p_k^{m_k}$ with $\gcd(d_i, d_j) = 1$ for $i\neq j$). The group generated by the single-qudit Clifford group and the intra-qudit gates; that is, the group
\[
G = \langle \mathcal{C}_d, \{CN_{k,k+1}\}_{k=1}^{n-1} \rangle,
\]
is dense in $\SU(d)$.
\end{theorem}
\begin{proof}
We verify the conditions of Corollary~\ref{cor:finite_groups}.

\textbf{1. Establishment of Local Irreducibility.}
First, we ensure that the local action on every factor is irreducible. Consider a subsystem $k$ of dimension $d_k$.
\begin{itemize}
    \item If $d_k = p$ is prime, by Theorem~\ref{thm:prime_irreducibility}, the local Clifford group $\mathcal{C}_{d_k}$ acts irreducibly on $\mathfrak{sl}(d_k, \mathbb{C})$.
    \item If $d_k = p^m$ is a prime power ($m \ge 2$), using Lemma~\ref{lem:induced_magic}, we obtain a local gate $T_{d_j}$. Since $\gcd(d_k, d_j)=1$, the order $s=d_j$ satisfies $s \nmid d_k$. Thus, $T_{d_j}$ satisfies the condition of Item 1 in Theorem~\ref{thm:prime_power_universality}, and consequently the local group $\langle \mathcal{C}_{d_k}, T_{d_j} \rangle$ acts irreducibly.
\end{itemize}
With irreducible local algebras established, the gates $\{CN_{k, k+1}\}$ serve as the mixing operators $V$. They do not normalize the product structure. By applying Theorem~\ref{thm:tensor_mixing}, the global group $G$ acts irreducibly on the global algebra $\mathfrak{sl}(d, \mathbb{C})$.

\textbf{2. Verification of Infiniteness.}
To prove universality, it remains only to show that $G$ is infinite. Since global irreducibility is guaranteed, it suffices to demonstrate that the group contains a dense subgroup on \emph{at least one} local subsystem. We identify such a subsystem by distinguishing two exhaustive cases based on the nature of the factorization of $d$:

\begin{itemize}
    \item Case A: The factorization contains at least one prime factor.
    Suppose there exists a subsystem $k$ such that $d_k = p$ is prime. Let $d_j$ be its coprime neighbor.
    We verify density on subsystem $k$. Using Proposition~\ref{prop:bezout_construction}, a local gate $T_S$ with $S = d_k d_j$ is constructed.
    Since $\gcd(d_k, d_j)=1$, we have $S \nmid d_k$ (and thus $S \nmid K_{d_k}$).
    By Corollary~\ref{cor:prime_universality}, the extended local group is dense in $\SU(p)$. Thus, the global group is infinite.

    \item Case B: The factorization consists entirely of composite prime powers.
    Suppose all factors are proper prime powers ($d_k = p_k^{m_k}$ with $m_k \ge 2$). By Theorem~\ref{thm:prime_power_universality}, the group is dense on factor $k$ if $S = d_k d_j$ satisfies the bound:
    \[ S > \pi(d_k-1)\big/(2\,\arcsin(1/4)). \]
    Consider the worst-case scenario with the smallest possible coprime proper powers: $d_k=4$ and $d_j=9$.
    The bound for $d_k=4$ is $S > \frac{3\pi}{2\,\arcsin(1/4)} \approx 18.65$ and $S = 4 \times 9 = 36$.
    Since $36 > 18.65$, the condition is strictly satisfied. Any other pair of proper powers yields a larger product, satisfying the bound strictly.
    Therefore, the group is dense in the subsystem $k$.
\end{itemize}

In both scenarios, we guarantee the existence of at least one locally dense subsystem. Combined with the global irreducibility established in the first part, this implies that $G$ is dense in the full group $\SU(d)$.
\end{proof}

\subsection{Explicit Universal Gate Sets}

A consequence of Theorem~\ref{thm:emergent_universality} is that, for composite systems with coprime dimensions, universality can be generated from Clifford gates together with internal arithmetic couplings.
In the qubit case ($d=2$), the group generated by local Clifford gates and CNOTs remains finite (the global Clifford group), so a further non-Clifford resource is needed to obtain density.
In contrast, for coprime composite systems, the inclusion of the intra-qudit gate $CN$---typically considered a standard permutation---is already sufficient, at the level of generating sets, to break the finite group structure. Due to the arithmetic mismatch between factors, the interplay between the local Clifford groups and the intra-qudit gate generates an infinite, dense group without taking an explicit diagonal magic gate as an additional primitive.

\begin{corollary}
Let $d = d_1 d_2 \cdots d_n$ be a proper composite dimension with pairwise coprime prime-power factors. The following finite set of gates constitutes a universal set for quantum computation on arrays of registers of dimension $d$:
\begin{equation}
    \mathcal{G} = \{ X, H_d, P \} \cup \bigcup_{k=1}^{n-1} \{ CN_{d_k, d_{k+1}} \} \cup \{ CN_{d,d} \}.
\end{equation}

\end{corollary}

\begin{remark}
The set $\mathcal{G}$ satisfies the universality condition of Theorem~\ref{thm:brylinski}. Note the distinction between the entangling resources: $CN_{d_k, d_{k+1}} \in \SU(d)$ is an intra-qudit operation that couples internal coprime factors and acts as the generator of non-Clifford resources. In contrast, $CN_{d,d} \in \SU(d^2)$ represents the canonical inter-qudit gate acting on two logical qudits.
\end{remark}

\section{Discussion and Outlook}\label{sec:discussion}

The results of this paper establish a trichotomy, for single-qudit universality in Clifford-based gate sets, determined by the prime factorization of the local Hilbert space qudit dimension $d$. In prime dimensions, the Clifford group is maximally rigid, so any non-Clifford gate suffices for universality. In prime-power dimensions $d=p^m$ with $m \geq 2$, not every non-Clifford gate is universal. One needs gates that couple the invariant sectors left disconnected by the Clifford action, and this can be achieved either by suitable members of the $T_s$ family ($d$-dimensional generalization of the qubit $T$ gate) or by simple permutations such as the transposition exchanging $\ket{0}$ and $\ket{1}$. In coprime composite dimensions, arithmetic couplings between the coprime factors, implemented by intra-qudit CNOTs, already provide the required non-Clifford resources and generate the corresponding diagonal phases associated with magic states.

One consequence is that for every $d \geq 4$, classical reversible permutations can serve as universal non-Clifford resources, although the appropriate permutation depends on the arithmetic structure of $d$. This should be distinguished from the prime-dimensional corollary for the specific transposition $(0\;1)$, which requires $p \ge 5$; for $p=2,3$ every permutation is affine and therefore Clifford (Remark~\ref{rem:non_affine_permutations}). In this sense, the higher-dimensional setting enlarges the range of gate types that can play the role of a universal resource.

The present work focuses on single-qudit universality, where the Clifford group is associated to the abelian group $\mathbb{Z}_d$. From single-qudit universality, general (multi-qudit) quantum universality follows via a result of Brylinski and Brylinski (Theorem~\ref{thm:brylinski}). An extension of this problem concerns multi-qudit Clifford groups associated to $\mathbb{Z}_d^n$, corresponding to $n$ qudits of dimension $d$. For $d$ prime, the maximality of the $n$-qudit Clifford group was established by Nebe, Rains, and Sloane~\cite{Nebe2001}, using $\ell$-adic methods and the classification of finite simple groups. An open problem is whether the trichotomy identified here extends to the multi-qudit setting for general $d$, and whether the structure of $\mathbb{Z}_d^n$ for composite $d$ introduces further phenomena beyond those captured by the single-qudit analysis presented here.

From a resource-theoretic perspective, magic and non-stabilizerness have already been studied extensively, including qubit and prime-dimensional qudit settings~\cite{Veitch2014, Howard2014, Howard2017, Howard2012, Campbell2012, Beverland_2020, Bravyi_2019}. Our contribution addresses a different question, however. 
We classify which local non-Clifford gates extend the Clifford group to a quantum universal set. From the fault-tolerant point of view, the standard Clifford~+~$T$/magic-state paradigm remains the main comparison point. Our results show that, in higher-dimensional settings, the universalizing resource need not always appear as an explicit diagonal $T$-type gate in the generating set: depending on $d$, the same structural role can be played by simple permutations or by arithmetic couplings between coprime factors.

On the experimental side, platforms accessing non-binary levels---trapped ions~\cite{RingbauerEtAl2022, Low2025}, superconducting circuits~\cite{Chicco_2023, Morvan2021}, and photonic systems~\cite{Wang2020review, liu2024}---provide settings for higher-dimensional quantum computation. In such platforms, diagonal phase control, basis permutations, and couplings between internal factors of a composite register may arise more or less readily depending on the hardware. Our results therefore identify which kinds of local operations are sufficient to leave the Clifford regime in different dimensional settings. Our results should be understood as a structural statement about universal generating sets; however, the choice of a specific non-Clifford set (CNOT-type or $T$-type) is architecture-dependent.

A direction for future work is to compare the implementation cost and fault-tolerant overhead associated with these higher-dimensional mechanisms against more standard low-dimensional approaches, such as magic state distillation~\cite{Litinski2019, Gidney_2021} or the recently proposed magic state cultivation~\cite{Gidney_2024}. Determining whether any of these alternative routes yields a practical advantage requires an architecture-specific analysis of control, noise, and compilation overhead. We leave such an analysis for future research.


\begin{acknowledgments}
The authors thank Vsevolod I. Yashin and Xingshan Cui for their valuable comments and suggestions. C.G.\ was partially supported by Grant INV-2025-213-3452 from the School of Science of Universidad de los Andes. J.R.\ acknowledges support from the Office of the Vice-president of Research and Creative Activities and the Office of the Faculty of Science Vice-president of Research of Universidad de los Andes under the FAPA grant.
\end{acknowledgments}

\bibliography{biblio}

\end{document}